\documentclass[a4paper,12pt]{article}
\pdfpagewidth
\paperwidth
\pdfpageheight
\paperheight
\usepackage[english]{babel} 
\usepackage{amsmath,amssymb}
\usepackage[T1]{fontenc} 
\usepackage[utf8]{inputenc} 
\usepackage[usenames,dvipsnames]{xcolor}
\usepackage{graphicx,color}
\usepackage{geometry}
\usepackage{caption}
\usepackage{amsthm}
\usepackage{hyperref}
\usepackage{comment}

\DeclareMathOperator{\syn}{syn}
\def\F{\mathbb F}
\def\Fq{\mathbb{F}_q}

\def\PG{{\rm{PG}}}
\def\cC{\mathcal{C}}

\theoremstyle{plain} 
\newtheorem{thm}{Theorem}[section] 
\newtheorem{cor}[thm]{Corollary} 
\newtheorem{lem}[thm]{Lemma} 
\newtheorem{prop}[thm]{Proposition} 

\newtheorem{proposition}[thm]{Proposition}

\theoremstyle{definition}
\newtheorem{defn}[thm]{Definition}

\newtheorem{definition}[thm]{Definition}
\theoremstyle{remark}
\newtheorem{remark}[thm]{Remark}

\title{Breaking linearity in the sum-rank metric: MSRD codes from switched $\sigma$-rational normal curves}
\date{}

\author{Daniele Bartoli\thanks{Dipartimento di Matematica e Informatica, Universit\`a degli studi di Perugia, Perugia, Italy. daniele.bartoli@unipg.it}\and
Giovanni Giuseppe Grimaldi\thanks{Dipartimento di Matematica e Informatica, Universit\`a degli studi di Perugia, Perugia, Italy. giovannigiuseppe.grimaldi@unipg.it}\and
Giovanni Longobardi\thanks{Dipartimento di Matematica e Applicazioni ``R. Caccioppoli", Universit\`a di Napoli Federico {II}, Napoli, Italy. giovanni.longobardi@unina.it}}

\begin{document}
\maketitle

\begin{abstract}
We construct a family of scalar-closed, non-additive maximum sum-rank
distance codes (MSRD) by switching norm components of a $\sigma$-rational normal curve
and lifting the resulting point set through a linearized Reed--Solomon
syndrome map. We characterize the admissible switches
by a multiplicative stability condition on the norm classes.
In particular, for $1 \leq n_i \leq m$, $N=n_1+\ldots+n_{\ell}$ and $2\leq\delta\leq N-1$, every partition of $\F_q^*$ satisfying a specific condition, referred to as Condition $(\diamond)$, yields
a code in $\F_{q^m}^{N}$ of size $q^{m(N-\delta+1)}$ and minimum sum-rank
distance $\delta$; the code is non-additive whenever a curve component is
retained. For $\ell\geq2$ these appear to be the first non-additive MSRD codes
in the literature, all previously known families with more than one block
being $\F_q$-linear. Moreover, every code of the family has the same sum-rank weight
distribution as an $\F_{q^m}$-linear MSRD code with the same parameters,
although the geometry of the switched set distinguishes it from
linearized Reed--Solomon codes. In the single-block case, an explicit rank isometry identifies
the construction with the cone codes of Durante, Grimaldi and Longobardi.
\end{abstract}

\noindent \\

\noindent 

\section{Introduction}

The sum-rank metric provides a common framework for the Hamming and
the rank metric, which it recovers as the two extreme cases of its block
structure. Introduced in the context of multishot network coding, it has
since become the natural setting for space-time codes and for locally
repairable and partial-MDS codes in distributed storage; for a deeper
introduction see \cite{Neri,MPSK}. Codes attaining the Singleton-like
bound in this metric are called maximum sum-rank distance, or MSRD,
codes, and they play the role that MDS codes play in the Hamming metric
and MRD codes play in the rank metric.

Since their introduction by Mart\'inez-Pe\~nas \cite{MartinezPenas},
MSRD codes have been investigated almost exclusively in the linear
setting. All the families known to us --- linearized Reed--Solomon codes
\cite{MartinezPenas} and their doubly extended version \cite{NSZ},
twisted linearized Reed--Solomon codes \cite{Neri}, the codes obtained
from extended Moore matrices \cite{MPExtMoore}, and the four
constructions of \cite{MP2024} --- are $\F_q$-linear, and most of them
are $\F_{q^m}$-linear. 

The same is true of the general theory: duality,
weight distributions and generalized weights \cite{BGR}, as well as the
geometric description of sum-rank metric codes by multiple $q$-systems
\cite{NSZ}, are all formulated for linear codes.

The situation in the rank metric, that is, for a single block, is
markedly different. There, non-additive MRD codes have been known for a
decade and form by now a well populated family
\cite{Coss,DonatiDurante,DGL,Sic,Otal}, and most of them arise from a
geometric construction: one starts from a point set of
$\PG(d,q^n)$ which is exterior with respect to a suitable secant variety,
and one \emph{switches} some of its components, replacing them by other
point sets which are again exterior.

To the best of our knowledge, no non-additive MSRD code with more
than one block has so far appeared in the literature. The only MSRD
construction we are aware of that formally admits non-linear ingredients
is the cartesian-product Construction 1 of \cite{MP2024}, whose component
codes are explicitly allowed to be ``linear or non-linear''; that
construction, however, leaves the number $\ell$ of blocks unchanged, so
that it turns non-additive MRD codes into non-additive MRD codes and can
produce a non-additive MSRD code with $\ell\geq2$ only if one is
available to begin with. The remaining constructions of \cite{MP2024}
require $\F_q$-linear or $\F_{q^m}$-linear inputs. The purpose of this
paper is to fill this gap: we exhibit a family of MSRD codes with more than one  block which are closed under multiplication by the
scalars of $\F_{q^m}$ but are \emph{not} additive, and which therefore
appear to be the first non-additive MSRD codes in a genuine sum-rank
setting. For $\ell=1$ our family specializes to the rank-metric codes of
\cite{DGL}, of which it is thus a sum-rank generalization.

In \cite{DonatiDurante}, a geometric construction of a family of
rectangular non-linear MRD codes was exhibited by Donati and Durante from maximum
exterior sets obtained by switching norm components of a $C_F^\sigma$-set
with suitable components on the line joining its two vertices, see also \cite{Coss,Sic}. Later in
\cite{DGL}, Durante, Grimaldi and Longobardi introduced square non-linear
MRD codes by taking cones over these exterior sets. More precisely, their
codes arise from projective cones whose base is a maximum exterior set
with respect to a suitable secant variety of a canonical subgeometry and
whose vertex is a complementary projective subspace. The earlier
rectangular non-linear MRD codes are obtained by puncturing the square
codes in their matrix representation. The aim of the present paper is to
extend this construction from the rank metric to the sum-rank metric.

Our starting point is a linearized Reed--Solomon code and its
associated multiple $q$-system, which lies on a $\sigma$-rational normal
curve \cite{ABCS}. This curve is partitioned into $q-1$ scattered linear
sets indexed by norm classes, see also \cite{DonatiDurante}.
Following the Donati--Durante switching construction, we replace some
curve components with the corresponding norm components on the line
joining the two distinct points of the curve.  The main difficulty here is that the exterior set has to be constructed with respect to a variety which is no longer the secant variety of a subgeometry, but rather the projective variety $\Omega_{\delta-2}(\mathcal U_\Lambda)$,
whose points arise as syndromes vectors of bounded sum-rank weight associated with the LRS multiple $q$-system. In particular, its combinatorial behaviour is governed by the set $\Lambda$ of norm classes supporting the blocks. We prove that the resulting point
set and all its secants avoid this low-weight syndrome variety if and
only if the original norm classes are replaced and the retained norm
classes satisfy a multiplicative stability condition.
Explicitly, writing $\F_q^*=S\,\dot\cup\,T$, where $T$ indexes the
components taken on the line and $S$ those kept on the curve, Theorem
\ref{thm:main-exterior} states that this happens if and only if
$\Lambda\subseteq T$ and $S\,\Pi_{\delta-1}(\Lambda)\subseteq S$, where
$\Pi_{\delta-1}(\Lambda)$ denotes the set of products of $\delta-1$
elements of $\Lambda$, with repetitions and with multiplicities at most
$m$. The second condition is purely group-theoretic: it says that $S$ is
a union of cosets of the multiplicative subgroup $G=\langle\Pi_{\delta-1}(\Lambda)\rangle$ of $\F_q^*$ such that it is disjoint from
$\Lambda G$.
The proof relies on the root-space bound for skew polynomials (Proposition \ref{prop:rootspace}) and, in the
equality case, on a relation between the extreme coefficients of a skew
polynomial and the norm classes of its root spaces, see Proposition \ref{prop:viete}. Finally, we
lift the switched exterior set by means of the syndrome map of the LRS
code. We show that the syndrome preimage is precisely an affine cone
whose vertex is the kernel of the syndrome map and whose base is
projectively equivalent to the switched exterior set. This yields a
family of scalar-closed MSRD codes that are non-additive whenever at
least one curve component is retained.

The codes so obtained are far from shapeless. In Theorem
\ref{thm:maximal-subcodes} we prove that their maximal
$\F_{q^m}$-linear subcodes are precisely the $q^m+1$ syndrome preimages
of the points of the switched set, that they all contain the kernel of
the syndrome map and that their pairwise intersections equal that kernel;
Theorem \ref{thm:DP-MSRD} determines which of them are themselves MSRD.
This yields a decomposition of the code (Remark \ref{rm:dec}) from which
we determine its weight distribution: in Theorem
\ref{thm:weight-invariance} we prove that $\mathcal C_{S,T}$ has exactly
the sum-rank weight distribution of an $\F_{q^m}$-linear MSRD code with
the same parameters, and we give the number of codewords of minimum
weight explicitly. Hence the weight distribution cannot distinguish the
non-additive codes from the linear ones. The decomposition does,
however, make the projective equivalence class of the switched set an
invariant of the code under $\F_{q^m}$-semilinear sum-rank equivalence
(Theorem \ref{cor:geometric-invariant}); in particular, no member of the
family with $S\neq\emptyset$ is equivalent to a linearized Reed--Solomon
code.

When there is a single block with $\Lambda=\{1\}$ and $a_1=1$, the
switched base is exactly the Donati--Durante exterior set, while the
syndrome preimage is, up to rank-metric equivalence, the square cone code
constructed by Durante, Grimaldi and Longobardi. We give an explicit
isometry establishing this equivalence. Thus the construction presented
here recovers the known rank-metric family and extends it to a sum-rank
setting.

The paper is organized as follows. Section \ref{DonDurcode} recalls
the rank-metric codes $\mathcal C_{\sigma,T}$ of \cite{DGL}. Section 3
introduces the variety $\Omega_r(\mathcal U)$ and relates it to
the sum-rank weight. Section \ref{geo:setting} contains the switching construction and
the characterization of the admissible partitions. Section \ref{thefamily:cst} is devoted
to the codes $\mathcal C_{S,T}$, to their maximal linear subcodes and to
their weight distribution. Section \ref{sec:rank-recovery} treats the
case $\ell=1$.

\section{Preliminaries and Notation}

Let $\F_{q^m}$ be a finite field with $q^m$ elements and let
\[ \bold{n}=(n_1,\ldots,n_\ell), \qquad N=n_1+\cdots+n_\ell.\]
The ambient space is written as
\[ \F_{q^m}^{\bold{n}} = \F_{q^m}^{n_1}\oplus\cdots\oplus\F_{q^m}^{n_\ell}. \]
For a vector
\[ \bold{z}=(\boldsymbol{z}^{(1)}\mid\cdots\mid \boldsymbol{z}^{(\ell)}) \in \F_{q^m}^{\bold{n}}, \qquad \boldsymbol{z}^{(i)}=(z^{(i)}_1,\ldots,z^{(i)}_{n_i}), \]
its \textit{sum-rank weight} is defined by
\[ \operatorname{wt}_{\rm SR}(\bold{z}) = \sum_{i=1}^{\ell} \dim_{\F_q} \langle z^{(i)}_1,\ldots,z^{(i)}_{n_i} \rangle_{\F_q}, \] where $\dim_{\F_q} \langle z^{(i)}_1,\ldots,z^{(i)}_{n_i}\rangle_{\F_q}$ denotes the \textit{$\F_q$-rank} of $\boldsymbol{z}^{(i)}$.

The corresponding distance is
\[ d_{\rm SR}(\bold{x},\bold{y}) = \operatorname{wt}_{\rm SR}(\bold{x}-\bold{y}). \]
If $\ell=1$, this is the rank metric, whereas if $n_1=\cdots=n_\ell=1$, it is the Hamming metric. Thus sum-rank metric codes simultaneously extend rank-metric codes and Hamming-metric codes. A code in the sum-rank metric is an arbitrary subset
$ \cC\subseteq\F_{q^m}^{\bold{n}}$
endowed with the sum-rank distance. The code $\mathcal{C}$ is said to be \textit{linear} if it is an $\F_{q^m}$-linear subspace. Two (linear) sum-rank metric codes  $\mathcal C,\mathcal C'\subseteq\F_{q^m}^{\bold n}$ are said to be \emph{equivalent} if there exists an ($\F_{q^m}$-linear) isometry $\varphi:\F_{q^m}^{\bold n}\longrightarrow\F_{q^m}^{\bold n}$ such that $\varphi(\mathcal{C})=\mathcal{C'}$. The following provides a characterization of $\F_{q^m}$-linear isometries.
\begin{thm}\cite[Theorem V2]{sal}
    Let $\varphi: \F_{q^m}^{\bold n} \longrightarrow \F_{q^m}^{\bold n}$ be an $\F_{q^m}$-linear isometry. Then there exist a permutation $\tau : \{1,\ldots, \ell\} \longrightarrow \{1,\ldots, \ell\}$ with the property that $\tau(i)=j$ implies $n_i=n_j$, $a_1,\ldots,a_{\ell} \in \F_{q^m}^*$ and invertible matrices $M_1 \in \mathrm{GL}(n_1,\F_q), \ldots, M_{\ell} \in \mathrm{GL}(n_{\ell},\F_q)$ such that
    $$
    \varphi(\boldsymbol{z}^{(1)}\mid\cdots\mid \boldsymbol{z}^{(\ell)})=\left(a_1 \boldsymbol{z}^{(\tau^{-1}(1))} M_1 \mid \cdots \mid a_{\ell}\boldsymbol{z}^{(\tau^{-1}(\ell))} M_{\ell}\right)
    $$
    for any $(\boldsymbol{z}^{(1)}\mid\cdots\mid \boldsymbol{z}^{(\ell)}) \in \F_{q^m}^{\bold n}$.
\end{thm}
 When $\cC$ is an $\F_{q^m}$-linear subspace of dimension $k$, its minimum sum-rank distance satisfies the Singleton-like bound
\[ d_{\rm SR}(\cC)\leq N-k+1. \]
More generally, an arbitrary code of minimum distance $d$ satisfies
\[ |\cC|\leq q^{m(N-d+1)}, \] see \cite[Theorem 1.4]{MPSK}.
 Codes attaining the Singleton-like bound are called \textit{maximum sum-rank distance} codes, or MSRD codes. One of the main features of the sum-rank metric is the existence of several equivalent algebraic representations. These representations generalize the standard descriptions of rank-metric codes by linearized polynomials, matrices, and vectors, see \cite{Neri}.

From now on, unless otherwise stated, we consider full blocks, so that $n_1=\cdots=n_\ell=m$ and $N=\ell m$. Let $\sigma:x\mapsto x^{q^s}$, $\gcd(s,m)=1$
be a generator of $\operatorname{Gal}(\F_{q^m}\vert \F_q)$, and let $\F_{q^m}[X;\sigma]$ be the skew-polynomial ring determined by
\[ Xc=\sigma(c)X, \qquad c\in\F_{q^m}. \]
For $a\in\F_{q^m}$, define the truncated norms
\[ N_0(a)=1, \qquad N_j(a)=a\sigma(a)\cdots\sigma^{j-1}(a) \quad (j\geq1). \]
In particular,
$ N_m(a)=\operatorname{N}_{q^m/q}(a)$, i.e., the norm of the element $a \in \F_{q^m}$ over $\F_q$.
Choose elements
$ a_1,\ldots,a_\ell\in\F_{q^m}^* $
whose norms
$ \lambda_i=\operatorname{N}_{q^m/q}(a_i) $
are pairwise distinct, and put 
$$\Lambda=\{\lambda_1,\ldots,\lambda_\ell\} \subseteq\F_q^*$$
Clearly, $1\leq\ell\leq q-1$. The central skew polynomial
\[ H_\Lambda(X) = \prod_{i=1}^{\ell}(X^m-\lambda_i) \]
defines the quotient algebra 
\[ \mathcal R_\Lambda = \F_{q^m}[X;\sigma]/(H_\Lambda(X)). \]
Since $\deg H_\Lambda=\ell m$, the quotient algebra $\mathcal R_\Lambda$ is an $\F_{q^m}$-vector space of dimension $\ell m$. Each element of $\mathcal R_\Lambda$ will be represented by its unique skew polynomial of degree less than $N$. For
\[ F(X)=\sum_{j\geq0} f_jX^j\in\F_{q^m}[X;\sigma], \]
where only finitely many coefficients $f_j$ are nonzero, the right evaluation at $t\in\F_{q^m}$ is
\begin{equation} \label{eva}
        F(t)=\sum_{j\geq0} f_j N_j(t).
\end{equation}
Define
\[ F_{a_i}(x) = \sum_j f_jN_j(a_i)\sigma^j(x). \]
Each $F_{a_i}$ is an $\F_q$-linear endomorphism of $\F_{q^m}$. Let $\boldsymbol a=(a_1,\ldots,a_{\ell})$ and consider the map
\[ \begin{split} \Phi_{\boldsymbol a}:\mathcal R_\Lambda &\longrightarrow \operatorname{End}_{\F_q}(\F_{q^m})^\ell,\\ F &\longmapsto \bigl(F_{a_1},\ldots,F_{a_\ell}\bigr) \end{split} \]
is an $\F_{q^m}$-linear isomorphism, see \cite[Theorem 4.1]{Neri}. Under this correspondence, the sum-rank weight of $F$ is defined as 
\[ \operatorname{wt}_{\rm SR}(F) = \sum_{i=1}^{\ell} \operatorname{rk}_{\F_q}(F_{a_i}). \]
Equivalently, if
$ d_{\lambda_i}(F) = \dim_{\F_q}\ker(F_{a_i}), $
then
$ \operatorname{wt}_{\rm SR}(F) = \ell m-\sum_{i=1}^{\ell}d_{\lambda_i}(F). $

Note that the value $d_\lambda$ makes sense for every $\lambda\in\F_q^*$, and not
only for the classes belonging to $\Lambda$. Indeed, for $a\in\F_{q^m}^*$ with
$\operatorname{N}_{q^m/q}(a)=\lambda$ we put
\[
 F_a(x)=\sum_{j\geq0}f_jN_j(a)\sigma^j(x),
 \qquad
 d_\lambda(F)=\dim_{\F_q}\ker(F_a).
\]

\begin{remark}\label{rem:well-defined}
The integer $d_\lambda(F)$ does not depend on the chosen representative
$a$. Indeed, if $\operatorname{N}_{q^m/q}(a)=\operatorname{N}_{q^m/q}(a')$,
then $\operatorname{N}_{q^m/q}(a'/a)=1$ and
there exists $w\in\F_{q^m}^*$ with $a'=a\,\sigma(w)/w$. Since $N_j(a')=N_j(a)\sigma^j(w)/w$ for every $j\geq0$, we get
\[
 F_{a'}(x)=\sum_{j\geq0}f_jN_j(a)\frac{\sigma^j(w)}{w}\sigma^j(x)
          =\frac1w\,F_a(wx).
\]
Therefore $\ker F_{a'}=w^{-1}\ker F_a$, and the two kernels have the same
$\F_q$-dimension.
\end{remark}

\begin{prop}\cite[Theorem 3.4]{Neri}\label{prop:rootspace}
Let $\Xi\subseteq\F_q^*$ be a set of pairwise distinct elements and let
$F\in\F_{q^m}[X;\sigma]$ with $F\neq0$. Then
\[
 \sum_{\lambda\in\Xi}d_{\lambda}(F) \leq \deg(F).
\]
\end{prop}

\begin{proposition}\cite[Proposition 6.1]{Neri}
\label{prop:viete}  Let $\Xi \subset \F_q^*$ be a set of pairwise
distinct elements and
$F(X)=f_0+\cdots+f_dX^d$ with $f_0f_d\neq0$.  If
\[
 d=\deg(F)=\sum_{\lambda\in\Xi}d_\lambda(F),
\]
then
\begin{equation}\label{eq:viete}
 N_{m}\left(\frac{f_0}{f_d}\right)
 =(-1)^{dm}\prod_{\lambda\in\Xi}\lambda^{d_\lambda(F)}.
\end{equation}
\end{proposition}

The same correspondence may be written using linearized polynomials. We use the notation $x^{\sigma^j}=\sigma^j(x)$; these polynomials are regarded as operators on $\F_{q^m}$. The skew monomial $f_jX^j$, evaluated in the $i$-th norm class, corresponds to the $\sigma$-linearized monomial
\[ f_jN_j(a_i)x^{\sigma^j}. \]
Thus a skew polynomial $F(X)=\sum_jf_jX^j$ determines the tuple
\[ \left( \sum_jf_jN_j(a_1)x^{\sigma^j}, \ldots, \sum_jf_jN_j(a_\ell)x^{\sigma^j} \right) \]
of $\sigma$-linearized polynomials. The $i$-th component represents the $\F_q$-linear endomorphism $F_{a_i}$. After fixing an $\F_q$-basis
$ \mathcal B_i=(\beta_{i1},\ldots,\beta_{im}) $
of $\F_{q^m}$ for each block, the endomorphism $F_{a_i}$ is represented by a matrix
$ M_{\mathcal B_i}(F_{a_i}) \in\F_q^{m\times m}. $
Consequently, the map
\[ F\longmapsto \bigl( M_{\mathcal B_1}(F_{a_1}), \ldots, M_{\mathcal B_\ell}(F_{a_\ell}) \bigr) \]
identifies $\mathcal R_\Lambda$ with a space of $\ell$-tuples of $m\times m$ matrices. In this representation,
\[ \operatorname{wt}_{\rm SR}(F) = \sum_{i=1}^{\ell} \operatorname{rk}_{\F_q} \bigl(M_{\mathcal B_i}(F_{a_i})\bigr). \]
Hence the sum-rank metric is the sum of the ordinary matrix ranks of the blocks. There is also a vector representation. Evaluating each endomorphism on the chosen basis gives
\[  \operatorname{ev}_{\boldsymbol a,\boldsymbol{\mathcal B}}(F) = \bigl( F_{a_1}(\beta_{11}),\ldots,F_{a_1}(\beta_{1m}) \mid\cdots\mid F_{a_\ell}(\beta_{\ell1}),\ldots, F_{a_\ell}(\beta_{\ell m}) \bigr) \in\F_{q^m}^{\ell m}.  \]
The $\F_q$-rank of the $i$-th vector block is exactly the rank of the endomorphism $F_{a_i}$, and hence
\[ \operatorname{wt}_{\rm SR} \bigl( \operatorname{ev}_{\boldsymbol a,\boldsymbol{\mathcal B}}(F) \bigr) = \operatorname{wt}_{\rm SR}(F). \]
We therefore have isometric descriptions
\[ \mathcal R_\Lambda \longleftrightarrow \operatorname{End}_{\F_q}(\F_{q^m})^\ell \longleftrightarrow \bigl(\F_q^{m\times m}\bigr)^\ell \longleftrightarrow \F_{q^m}^{\ell m}. \]
In other words, a sum-rank code can be regarded equivalently as a set of skew polynomials modulo $H_\Lambda$, a set of tuples of $\sigma$-linearized polynomials, a set of tuples of matrices over $\F_q$, or a set of block vectors over $\F_{q^m}$. The most well-known linear MSRD codes are the linearized Reed--Solomon codes introduced by Mart\'inez-Pe\~nas \cite{MartinezPenas} . In the skew-polynomial framework, the code of dimension $k$ ($1\leq k\leq N$) is simply
\[ \operatorname{LRS}_k(\Lambda) = \{F\in\mathcal R_\Lambda:\deg(F)<k\}. \]
Proposition \ref{prop:rootspace} implies that every nonzero $F\in\operatorname{LRS}_k(\Lambda)$ satisfies
\[ \operatorname{wt}_{\rm SR}(F) \geq \ell m-(k-1) = \ell m-k+1. \]
Therefore,
\[ d_{\rm SR} \bigl(\operatorname{LRS}_k(\Lambda)\bigr) = \ell m-k+1, \]
and linearized Reed--Solomon codes are MSRD. They simultaneously generalize Reed--Solomon codes and generalized Gabidulin codes: in the full-block model, $m=1$ recovers Reed--Solomon codes, while a single full-block recovers generalized Gabidulin codes.

\subsection{The dual of a sum-rank metric code}

Given
\[
 F=\sum_{i=0}^{N-1}f_iX^i,
 \qquad
 G=\sum_{i=0}^{N-1}g_iX^i
\]
in $\mathcal R_\Lambda$, define
\begin{equation}
 \langle F,G\rangle_\Lambda
 =\mathrm{Tr}_{q^m/q}\left(\sum_{i=0}^{N-1}f_ig_i\right).
 \label{eq:Lambda-form}
\end{equation}
where $\mathrm{Tr}(\cdot)$ denotes the trace map from $\F_{q^m}$ to $\F_q$. This is a non-degenerate $\F_q$-bilinear form.

\begin{definition}
Let $\mathcal C\subseteq\mathcal R_\Lambda$ be a $\F_q$-linear sum-rank metric
code. Its dual is
\[
 \mathcal C^\perp
 =\{G\in\mathcal R_\Lambda:
       \langle F,G\rangle_\Lambda=0
       \text{ for every }F\in\mathcal C\}.
\]
\end{definition}

Put $\boldsymbol a=(a_1,\ldots,a_\ell)$ and
$\boldsymbol a^{-1}=(a_1^{-1},\ldots,a_\ell^{-1})$ and assume that
$\Lambda$ is a subgroup of $\F_q^*$. Clearly, the norms of the entries of
$\boldsymbol a^{-1}$ then again form $\Lambda$.
Write the reduced operator representations as
\[
 F_{a_i}(x)=\sum_{r=0}^{m-1}c_{i,r}(F)x^{\sigma^r},
 \qquad
 G_{a_i^{-1}}(x)=\sum_{r=0}^{m-1}c^-_{i,r}(G)x^{\sigma^r},
\]
where
\begin{align*}
 c_{i,r}(F)
 &=N_r(a_i)\sum_{t=0}^{\ell-1}f_{r+tm}\lambda_i^t,\\
 c^-_{i,r}(G)
 &=N_r(a_i^{-1})\sum_{t=0}^{\ell-1}g_{r+tm}\lambda_i^{-t},
 \qquad 0\leq r<m.
\end{align*}
These expressions follow from $\sigma^m=\mathrm{id}$ and
$N_{r+tm}(a_i)=N_r(a_i)\lambda_i^t$.
By \cite[Proposition 5.3]{Neri},
\begin{equation}\label{scalarprod}
 \ell\langle F,G\rangle_\Lambda
 =\left\langle
   \Phi_{\boldsymbol a}(F),
   \Phi_{\boldsymbol a^{-1}}(G)
  \right\rangle_{\mathrm{srk}},
\end{equation}
where
\[
 \left\langle
   \Phi_{\boldsymbol a}(F),
   \Phi_{\boldsymbol a^{-1}}(G)
  \right\rangle_{\mathrm{srk}}
 =\sum_{i=1}^{\ell}\mathrm{Tr}_{q^m/q}
   \left(\sum_{r=0}^{m-1}c_{i,r}(F)c^-_{i,r}(G)\right).
\]
Indeed, the factors $N_r(a_i)N_r(a_i^{-1})$ equal $1$, and
$\sum_{i=1}^{\ell}\lambda_i^{t-u}$ is $\ell$ for $t=u$ and
$0$ otherwise, for $0\leq t,u<\ell$.

There is an analogous statement in the vector model, still under the subgroup hypothesis; see \cite[Theorem 5.5]{Neri}. Let
$\mathfrak B=(\mathcal{B}_1,\ldots,\mathcal{B}_\ell)$ be a collection of $\F_q$-bases of $\F_{q^m}$,
and let
$\mathfrak B^*=(\mathcal{B}_1^*,\ldots,\mathcal{B}_\ell^*)$ where $\mathcal{B}_j^*=(\beta_{j1}^*, \ldots, \beta_{jm}^*)$ is the corresponding dual basis of $\mathcal{B}_j=(\beta_{j1}, \ldots, \beta_{jm})$, that is
$$
\mathrm{Tr}_{q^m/q}(\beta_{jk}\beta_{jh}^*)=\begin{cases} 1 \qquad \mathrm{if} \,\, k=h  \\ 0 \qquad \mathrm{if} \,\, k \ne h. \end{cases}
$$
If $\mathcal C(\boldsymbol a,\mathfrak B)$ denotes the usual
vector representation of $\mathcal C$, then
\begin{equation}
 \mathcal C^\perp(\boldsymbol a,\mathfrak B)
 =\bigl(
    \mathcal C(\boldsymbol a^{-1},\mathfrak B^*)
   \bigr)^{\perp\!\!\!\!\perp},
 \label{vector-dual}
\end{equation}
where $\perp\!\!\!\!\perp$ is orthogonality with respect to the bilinear form on $\F_{q^m}^N$ defined as $\langle \mathbf{u},\mathbf{v} \rangle=\mathrm{Tr}_{q^m/q}(\mathbf{u}\mathbf {v}^{\top})$. 
Moreover, by \cite[Theorem 5.7]{Neri} if $\Lambda$ is a subgroup of $\F_q^*$ and $\mathcal C\subseteq\mathcal R_\Lambda$ is an $\F_q$-linear MSRD code, then
$\mathcal C^\perp$ is also an $\F_q$-linear MSRD code.

We now allow $\Lambda$ to be any set of distinct nonzero norms again.
For $1\leq k\leq N$, the linearized Reed--Solomon code of $\F_{q^m}$-dimension $k$ in
the skew-polynomial model is
\begin{equation}
 \mathrm{LRS}_k
 =\{F\in\mathcal R_\Lambda:\deg F<k\}
 =\langle1,X,\ldots,X^{k-1}\rangle_{\F_{q^m}}.
 \label{eq:LRS-definition}
\end{equation}
It has minimum sum-rank distance
$N-k+1$. We also put $\mathrm{LRS}_0=\{0\}$. The dual is given by
\begin{equation}
 (\mathrm{LRS}_k)^\perp
 =\langle X^k,X^{k+1},\ldots,X^{N-1}\rangle_{\F_{q^m}}
 =\mathrm{LRS}_{N-k} X^k
 =X^k \mathrm{LRS}_{N-k},
\label{dualLRS}
\end{equation}
see \cite[Theorem 5.9]{Neri}. Formula \eqref{dualLRS} is valid for an arbitrary set $\Lambda$ of pairwise distinct elements of $\F_q^*$. Indeed, orthogonality to every scalar multiple of $1,X,\ldots,X^{k-1}$ is equivalent to the vanishing of the first $k$ coefficients. Moreover, multiplication by $X^k$ induces composition with invertible operators in each nonzero norm class and hence preserves sum-rank weight.
In particular, $(\mathrm{LRS}_k)^\perp$ is equivalent, under a
sum-rank isometry, to the linearized Reed--Solomon code
$\mathrm{LRS}_{N-k}$.

\begin{remark}
Under the subgroup hypothesis, the compatibility above identifies coefficient-duality with vector duality. For $\ell=1$ and $\Lambda=\{1\}$, formula \eqref{dualLRS} specializes to the familiar
statement that the Delsarte dual of a Gabidulin code is equivalent to a
Gabidulin code of complementary dimension. For $m=1$, it specializes to the
classical duality of Reed--Solomon codes.
\end{remark}

\section{The non-linear MRD codes \texorpdfstring{$\mathcal{C}_{\sigma,T}$}{C sigma,T}}\label{DonDurcode}

In this section we recall the square rank-metric construction of \cite{DGL}. Let $m\geq 3$, let $2\leq k\leq m-1$, and put $d=m-k+1.$
Let
$\sigma\colon x\longmapsto x^{q^s}$
be a generator of
$\operatorname{Gal}(\mathbb{F}_{q^m}\vert \mathbb{F}_q)$, where
$\gcd(s,m)=1$, and let $1\in T\subseteq\mathbb{F}_q^*$.
For every $a\in\mathbb{F}_q^*$, set
\[
\mathcal{C}_a
 =
 \left\{
 x\in\mathbb{F}_{q^m}^*
 :
 \operatorname{N}_{q^m/q}(x)=a
 \right\}.
\]

Consider the canonical subgeometry
\[
\Sigma_{m,m}
 =
 \left\{
 (x,x^\sigma,\ldots,x^{\sigma^{m-1}})
 :
 x\in\mathbb{F}_{q^m}^*
 \right\}
 \cong \operatorname{PG}(m-1,q)
\]
of $\operatorname{PG}(m-1,q^m)$, and the two disjoint subspaces
\[
\Lambda^\star\colon
X_0=X_1=\cdots=X_{d}=0
\]
and
\[
\Lambda\colon
X_{d+1}=X_{d+2}=\cdots=X_{m-1}=0.
\]
Thus, $\Lambda^\star$ and $\Lambda$ have vector dimension $k-2$ and
$d+1=m-k+2$, respectively.

The projection of $\Sigma_{m,m}$ from
$\Lambda^\star$ onto $\Lambda$ is
\[
\Gamma
 =
 \left\{
 (x,x^\sigma,\ldots,x^{\sigma^d},0,\ldots,0)
 :
 x\in\mathbb{F}_{q^m}^*
 \right\}.
\]
This linear set is a $d$-embedding of $\Sigma_{m,m}$, see \cite[Theorem 6]{Lun2}.

Inside $\Lambda$, consider the two points
\[
A=(\underbrace{0,\ldots,0}_{d},1,0,\ldots,0),
\qquad
B=(1,0,\ldots,0),
\]
and, for every $a\in\mathbb{F}_q^*$, define
\[
\mathcal{X}_a
 =
 \left\{
 \left(
 1,t,t^{\sigma+1},\ldots,
 t^{\sigma^{d-1}+\cdots+\sigma+1},
 0,\ldots,0
 \right)
 :
 t\in \mathcal{C}_a
 \right\}  \subseteq \Lambda.
\]
The set
\[
\mathcal{X}
 =
 \{A,B\}\cup
 \bigcup_{a\in\mathbb{F}_q^*}\mathcal{X}_a
\]
is a $C_F^\sigma$-set with vertices $A$ and $B$. The subsets
$\mathcal{X}_a$ are its components.

For every $a\in\mathbb{F}_q^*$, define also the subset of the line
$\langle A,B\rangle$
\[
\mathcal{J}_a
 =
 \left\{
 \left(
 1,\underbrace{0,\ldots,0}_{d-1},
 (-1)^{d-1}t,
 0,\ldots,0
 \right)
 :
 t\in \mathcal{C}_a
 \right\}.
\]
The point set associated with $T$ is
\begin{equation}
\label{eq:ET}
\mathcal{E}_T
 =
 \left(
 \mathcal{X}\setminus
 \bigcup_{a\in T}\mathcal{X}_a
 \right)
 \cup
 \bigcup_{a\in T}\mathcal{J}_a.
\end{equation}

The set $\mathcal{E}_T$ has $q^m+1$ points and is a maximum exterior
set with respect to the $(d-2)$-secant variety
$\Omega_{d-2}(\Gamma)$.

Let
\[
\mathbb{K}_T
 =
 \mathbb{K}(\Lambda^\star,\mathcal{E}_T)
\]
be the cone with vertex $\Lambda^\star$ and base $\mathcal{E}_T$.
The cone $\mathbb{K}_T$ is a maximum exterior set with respect to
$\Omega_{d-2}(\Sigma_{m,m})$ and has
\[
|\mathbb{K}_T|
 =
 \frac{q^{mk}-1}{q^m-1}
\]
points.

We can now define the corresponding rank-metric code. For
\[
\boldsymbol{\alpha}
 =
 (\alpha_0,\ldots,\alpha_{m-1})
 \in\mathbb{F}_{q^m}^m,
\]
denote by
\[
f_{\boldsymbol{\alpha}}(X)
 =
 \sum_{i=0}^{m-1}\alpha_iX^{\sigma^i}
\]
the associated $\sigma$-linearized polynomial. Then
\begin{equation}
\label{eq:CsigmaT-geometric}
\mathcal{C}_{\sigma,T}
 =
 \{0\}
 \cup
 \left\{
 f_{\boldsymbol{\alpha}}(X)
 :
 \langle\boldsymbol{\alpha}\rangle_{\mathbb{F}_{q^m}}
       \in\mathbb{K}_T
 \right\}.
\end{equation}
Here, for each point of $\mathbb{K}_T$, all its nonzero vector
representatives are taken. Thus, every point of $\mathbb{K}_T$
contributes exactly $q^m-1$ nonzero linearized polynomials to
$\mathcal{C}_{\sigma,T}$. Consequently,
\[
|\mathcal{C}_{\sigma,T}|
 =
 1+(q^m-1)|\mathbb{K}_T|
 =
 q^{mk}.
\]

More explicitly, since the vertex $\Lambda^\star$ consists of the
points whose first $d+1$ coordinates vanish, the code is
\begin{align}
\label{eq:CsigmaT-explicit}
\mathcal{C}_{\sigma,T}
={}&
\left\{
 \sum_{i=0}^{d}
 \lambda\alpha^{\sigma^i}
 N_i(\xi)X^{\sigma^i}
 +
 \sum_{i=d+1}^{m-1}\beta_iX^{\sigma^i}
 :
 \begin{array}{l}
 \lambda,\alpha,\beta_i\in\mathbb{F}_{q^m},\\
 \xi\in\mathbb{F}_{q^m}^*,\\
 \operatorname{N}_{q^m/q}(\xi)\in\mathbb{F}_q^*\setminus T
 \end{array}
\right\}
\nonumber\\
&\cup
\left\{
 \lambda\alpha X
 +
 (-1)^{d+1}\lambda\alpha^\sigma\eta X^{\sigma^d}
 +
 \sum_{i=d+1}^{m-1}\beta_iX^{\sigma^i}
 :
 \begin{array}{l}
 \lambda,\alpha,\beta_i\in\mathbb{F}_{q^m},\\
 \eta\in\mathbb{F}_{q^m}^*,\\\operatorname{N}_{q^m/q}(\eta)\in T
 \end{array}
\right\}
\nonumber\\
&\cup
\left\{
 \alpha X^{\sigma^d}
 +
 \sum_{i=d+1}^{m-1}\beta_iX^{\sigma^i}
 :
 \alpha,\beta_i\in\mathbb{F}_{q^m}
\right\}
\nonumber\\
&\cup
\left\{
 \alpha X
 +
 \sum_{i=d+1}^{m-1}\beta_iX^{\sigma^i}
 :
 \alpha,\beta_i\in\mathbb{F}_{q^m}
\right\}.
\end{align}

Since $\mathbb{K}_T$ is a maximum exterior set with respect to
$\Omega_{d-2}(\Sigma_{m,m})$, the code
$\mathcal{C}_{\sigma,T}$ has minimum rank distance $d$ and cardinality
$q^{m(m-d+1)}=q^{mk}$. Hence, it is an $(m,m,q;d)\text{-MRD code}$.
For $q>2$ and
$1\in T\subsetneq\mathbb{F}_q^{*}$, the code is non-linear. Moreover,
$\mathcal{C}_{\sigma,T}$ contains two distinct subcodes equivalent to
the generalized Gabidulin code $\mathcal{G}_{k-1,\sigma}=\{ \sum_{i=0}^{k-2}f_iX^{\sigma^i} : f_i \in \F_{q^m}\}$. This
property distinguishes it from the non-linear Otal--\"Ozbudak codes constructed in \cite{Otal},
which contain a unique subcode equivalent to
$\mathcal{G}_{k-1,\sigma}$. Consequently,
$\mathcal{C}_{\sigma,T}$ is neither equivalent nor adjointly equivalent
to the non-linear Otal--\"Ozbudak family.

The construction also extends earlier geometric examples. Deleting the
last $k-2$ rows gives the rectangular MRD codes of Donati and Durante \cite{DonatiDurante},
while the cases $k=2$ and $(m,k)=(3,2)$ recover the constructions of
Durante--Siciliano \cite{Sic} and Cossidente--Marino--Pavese \cite{Coss}, respectively. If
$q=2$ or $T=\mathbb{F}_q^{*}$, the construction reduces, up to
equivalence, to the generalized Gabidulin code
$\mathcal{G}_{k,\sigma}$. These assertions, including the
inequivalence to the Otal--\"Ozbudak family and the puncturing relation,
are established in \cite{DGL}.

\section{The syndrome map}

In this section, let $n_i \geq 1$ for any $i \in \{1,\dots,\ell \}$ and $N=n_1 + \ldots + n_{\ell}$.
Thus a vector $\bold{z}\in\F_{q^m}^N$ is written as
$\bold{z}=(\boldsymbol{z}^{(1)}\mid\cdots\mid \boldsymbol{z}^{(\ell)})$, with $\boldsymbol{z}^{(i)}\in\F_{q^m}^{n_i}$, and
\[
\mathrm{wt}_{\mathrm{SR}} (\bold{z})=\sum_{i=1}^{\ell}
 \dim_{\F_q}\langle z^{(i)}_1,\ldots,z^{(i)}_{n_i}\rangle_{\F_q}.
\]

Fix an integer $2\leq \delta\leq N-1$.
Let $\bold{H}=(H_1\mid\cdots\mid H_\ell)\in\F_{q^m}^{(\delta+1)\times N}$ be a
full-rank matrix, with $H_i\in\F_{q^m}^{(\delta+1)\times n_i}$, and define
\[
 U_i=\operatorname{colsp}_{\F_q}(H_i)\leq\F_{q^m}^{\delta+1}.
\]

\begin{definition}
For $r\geq0$, define $\Omega_r(\mathcal{U})$, where
$\mathcal{U}=(U_1,\ldots,U_\ell)$, to be the subset of
$\PG(\delta,q^m)$ given by
\[
 \Omega_r(\mathcal{U})=
 \bigcup_{\substack{W_i\leq_{\F_q}U_i\\
                     \sum_i\dim_{\F_q}W_i\leq r+1}}
 \mathbb{P}\! \left(\langle W_1,\ldots,W_\ell\rangle_{\F_{q^m}}\right).
\]

\end{definition}
\medskip

\noindent Here, we adopt the convention: if $\mathcal{A}$ is a subset of $\F_{q^m}^{\delta+1}$, we denote by $\mathbb{P}(\mathcal{A})$ the set of all projective points whose representatives are non-zero vectors of $\mathcal{A}$. Moreover, 
$\mathbb{P}(\emptyset)=\mathbb P(\{\bold{0}\})=\emptyset$ and $\langle \bold{v} \rangle=\mathbb{P}(\{\bold{v}\})=\mathbb{P}(\langle \bold{v} \rangle_{\F_{q^m}})$ for a nonzero vector $\bold{v} \in \F_{q^m}^{\delta +1}$. 

\begin{remark} 
Note that $\Omega_{0}(\mathcal{U})=\cup_{i=1}^{\ell} \mathbb{P}( U_i )$, $\Omega_{\delta}(\mathcal{U})=\PG(\delta,q^m)$  and for any $1 \leq r \leq \delta$ we have $\Omega_{r-1}(\mathcal{U}) \subseteq \Omega_{r}(\mathcal{U})$. If $\ell=1$ and $\mathcal{U}=U_1$, the set $\Omega_r(\mathcal{U})$ is exactly the \textit{$r$-secant variety} of $\mathbb{P}(U_1 )$ defined in \cite{DGL}. 
\end{remark}
\medskip
Let
\[ \operatorname{syn}_\bold{H}:\bold{z} \in \F_{q^m}^N\longrightarrow \bold{H}\bold{z}^{\mathsf T} \in \F_{q^m}^{\delta+1},
\]
be the syndrome map. Let $K$ be the kernel of $\operatorname{syn}_\bold{H}$ and let $L$ be an $\F_{q^m}$-subspace of $\F_{q^m}^N$ such that
$\F_{q^m}^N=K \oplus L.
$

\begin{lem}\label{lem:syndrome}
For a vector $\bold{z}\in \F_{q^m}^N \setminus K$ and an integer $r\geq1$, one has
\[
 \mathrm{wt}_{\mathrm{SR}}(\bold{z})\leq r
 \quad\Longrightarrow\quad
 \langle\operatorname{syn}_\bold{H}(\bold{z})\rangle\in\Omega_{r-1}(\mathcal{U}).
\]
Conversely, every point of $\Omega_{r-1}(\mathcal{U})$ is represented by the syndrome of
some vector of $\F_{q^m}^N \setminus K$ with sum-rank weight at most $r$.
\end{lem}

\begin{proof} Let $\bold{z}=(\boldsymbol{z}^{(1)}\mid\cdots\mid \boldsymbol{z}^{(\ell)}) \in \F_{q^m}^N\setminus K$ be such that $\mathrm{wt}_{\mathrm{SR}}(\bold{z})\leq r$. For every $i\in\{1,\ldots,\ell\}$, write
\[ r_i= \dim_{\F_q} \langle z^{(i)}_1,\ldots,z^{(i)}_{n_i} \rangle_{\F_q}. \]
Then \[ \sum_{i=1}^{\ell}r_i=\mathrm{wt}_{\mathrm{SR}}(\bold{z})\leq r. \] Choose an $\F_q$-basis $(\alpha_{i1},\ldots,\alpha_{ir_i})$ of $ \langle z^{(i)}_1,\ldots,z^{(i)}_{n_i} \rangle_{\F_q}$. There exist $c_{kj}^{(i)}\in\F_q$ such that \[ z^{(i)}_k = \sum_{j=1}^{r_i}c_{kj}^{(i)}\alpha_{ij}, \qquad i=1,\ldots, \ell \text{ and } k=1,\ldots,n_i . \] Therefore, \[ \boldsymbol{z}^{(i)} = \sum_{j=1}^{r_i}\alpha_{ij}\boldsymbol{c}_j^{(i)}, \] where $\boldsymbol{c}_j^{(i)} = \left( c_{1j}^{(i)},\ldots,c_{n_i j}^{(i)} \right) \in\F_q^{n_i}$. It follows that \[ H_i(\boldsymbol{z}^{(i)})^{\mathsf T} = H_i \left( \sum_{j=1}^{r_i} \alpha_{ij}\boldsymbol{c}_j^{(i)} \right)^{\mathsf T} = \sum_{j=1}^{r_i} \alpha_{ij}H_i(\boldsymbol{c}_j^{(i)})^{\mathsf T}. \] Since $ H_i(\boldsymbol{c}_j^{(i)})^{\mathsf T}\in U_i, $ put 
$$ W_i = \left\langle H_i(\boldsymbol{c}_j^{(i)})^{\mathsf T} : j=1,\ldots,r_i \right\rangle_{\F_q}.$$ Then $W_i\leq_{\F_q}U_i$, 
$\dim_{\F_q}W_i\leq r_i$ and \[ H_i(\boldsymbol{z}^{(i)})^{\mathsf T} \in \langle W_i\rangle_{\F_{q^m}}. \] Consequently, \[ \operatorname{syn}_H(\bold{z}) = \bold{H}\bold{z}^{\mathsf T} = \sum_{i=1}^{\ell}H_i(\boldsymbol{z}^{(i)})^{\mathsf T} \in \langle W_1,\ldots,W_\ell\rangle_{\F_{q^m}}. \] Moreover, \[ \sum_{i=1}^{\ell}\dim_{\F_q}W_i \leq \sum_{i=1}^{\ell}r_i \leq r. \] Since $\bold{z}\notin K$, we have $\operatorname{syn}_{\bold{H}}(\bold{z})\neq0$, and hence \[ \langle\operatorname{syn}_H(\bold{z})\rangle \in \Omega_{r-1}(\mathcal U). \] Conversely, let $\langle \bold{y} \rangle \in\Omega_{r-1}(\mathcal U)$. Then there exist $\F_q$-subspaces \[ W_i\leq_{\F_q}U_i, \qquad i=1,\ldots,\ell, \] such that \[ \sum_{i=1}^{\ell}\dim_{\F_q}W_i\leq r \] and $\bold{y}\in \langle W_1,\ldots,W_\ell\rangle_{\F_{q^m}}$. Let $ r_i=\dim_{\F_q}W_i$ and let $\left  (\boldsymbol{w}^{(i)}_{1},\ldots,\boldsymbol{w}^{(i)}_{r_i}\right )$ be an $\F_q$-basis of $W_i$. 
Since $\bold{y}$ belongs to the $\F_{q^m}$-span of $W_1,\ldots,W_\ell$, there exist $\gamma_{ij}\in\F_{q^m}$ such that \[ \bold{y} = \sum_{i=1}^{\ell} \sum_{j=1}^{r_i} \gamma_{ij}\boldsymbol{w}^{(i)}_{j}. \] Since $\boldsymbol{w}^{(i)}_{j}\in U_i$, there exists $\boldsymbol{c}_{j}^{(i)}\in\F_q^{n_i}$ such that \[ \boldsymbol{w}_{j}^{(i)}=H_i\left (\boldsymbol{c}_{j}^{(i)}\right )^{\mathsf T}. \] Hence \[ \bold{y} = \sum_{i=1}^{\ell} H_i \left( \sum_{j=1}^{r_i}\gamma_{ij}\boldsymbol{c}^{(i)}_{j} \right)^{\mathsf T}. \] For every $i$, put \[ \boldsymbol{z}^{(i)} = \sum_{j=1}^{r_i}\gamma_{ij}\boldsymbol{c}_{j}^{(i)} \in\F_{q^m}^{n_i} \] and define \[ \bold{z}=(\boldsymbol{z}^{(1)}\mid\cdots\mid \boldsymbol{z}^{(\ell)}) \in\F_{q^m}^N. \] Then \[ \operatorname{syn}_{\bold{H}}(\bold{z}) = \bold{H}\bold{z}^{\mathsf T} = \bold{y}. \] 
Since $\bold{y}\neq0$, it follows that $\bold{z}\notin K$. Write \[ \boldsymbol{c}^{(i)}_{j} = \left( c_{j1}^{(i)},\ldots,c_{jn_i}^{(i)} \right), \qquad c_{jk}^{(i)}\in\F_q. \] For every $k\in\{1,\ldots,n_i\}$, we have $ z^{(i)}_k = \sum_{j=1}^{r_i} \gamma_{ij}c_{jk}^{(i)}. $ Therefore, $ z^{(i)}_k \in \left\langle \gamma_{i1},\ldots,\gamma_{ir_i} \right\rangle_{\F_q}, $ and hence \[ \left\langle z^{(i)}_1,\ldots,z^{(i)}_{n_i} \right\rangle_{\F_q} \leq \left\langle \gamma_{i1},\ldots,\gamma_{ir_i} \right\rangle_{\F_q}. \] It follows that \[ \dim_{\F_q} \left\langle z^{(i)}_1,\ldots,z^{(i)}_{n_i} \right\rangle_{\F_q} \leq r_i. \] Consequently, \[ \mathrm{wt}_{\mathrm{SR}}(\bold{z}) \leq \sum_{i=1}^{\ell}r_i = \sum_{i=1}^{\ell}\dim_{\F_q}W_i \leq r. \] The claim follows.
\end{proof}

\begin{defn}
A set $\mathcal E\subseteq\PG(\delta,q^m)$ is an \textit{exterior set} with respect to
$\Omega_{r}(\mathcal{U})$, $0\leq r\leq\delta$, if every line joining two distinct points of
$\mathcal E$ is disjoint from $\Omega_{r}(\mathcal{U})$. 
\end{defn}

\section{The geometric setting}\label{geo:setting}

Throughout this section, let $1 \leq n_i \leq m$, $i \in \{1,\ldots,\ell \}$, put $N=n_1+\ldots +n_{\ell}$ and assume $2\leq\delta\leq N-1$. Fix $\sigma:x \mapsto x^{q^s},\gcd(s,m)=1,$
and let $\F_{q^m}[X;\sigma]$ be the skew-polynomial ring. Let $\Lambda=\{\lambda_1,\ldots,\lambda_\ell\}\subseteq\F_q^*$ consisting of pairwise distinct elements. Choose $a_i\in\F_{q^m}^*$ with
$N_m(a_i)=\lambda_i$, and choose an $\F_q$-basis
$(\beta_{i1},\ldots,\beta_{im})$ of $\F_{q^m}$ for each $i\in \{1,\ldots,\ell\}$.  Define
\[
 \boldsymbol{v}_{a_i}(x)=
 \bigl(x,N_1(a_i)\sigma(x),\ldots,N_\delta(a_i)\sigma^\delta(x)\bigr)
 \in\F_{q^m}^{\delta+1}
\]
and
\begin{equation}\label{H}
 H_i=
 \bigl(\boldsymbol{v}^{\mathsf T}_{a_i}(\beta_{i1}), \ldots, \boldsymbol{v}^{\mathsf T}_{a_i}(\beta_{in_i})\bigr),
 \qquad
 \bold{H}=(H_1\mid\cdots\mid H_\ell) \in \F_{q^m}^{(\delta+1) \times N} .
\end{equation}
The row space of $\bold{H}$ is an LRS code of dimension $\delta+1$ and length
$N$, hence it is MSRD, see \cite{MartinezPenas, Neri}.  In particular, $\bold{H}$ has
full row rank.  In the following we will consider
\[
        \mathcal{U}_\Lambda=(U_{\lambda_1},\ldots,U_{\lambda_\ell}),
        \qquad U_{\lambda_i}=\operatorname{colsp}_{\F_q}(H_i).
\]
Note that, for any $i \in \{1,\ldots,\ell\}$, the map $\boldsymbol{v}_{a_i}: x \in \F_{q^m} \mapsto \boldsymbol{v}_{a_i}(x) \in U_{\lambda_i}$ is an $\F_q$-linear isomorphism.

In $\PG(\delta,q^m)$ consider the $\sigma$-rational normal curve
\[
 \mathcal{X}=\{P(t):t\in\F_{q^m}\}\cup\{A\},
 \qquad
 P(t)=(1,t,N_2(t),\ldots,N_\delta(t)),
\]
where
\begin{equation}\label{puntiAB}
 A= \langle \bold{e}_\delta \rangle ,\qquad B=P(0)=\langle \bold{e}_0 \rangle,
\end{equation}
where $ \bold{e}_i$ is the $i$-th standard vector $\F_{q^m}^{\delta+1}$, $i=0,\ldots,\delta$. The rational points of $\mathcal{X}$ are partitioned as
\[
 \mathcal{X}=\{A,B\}\ \dot\cup\!
       \bigcup_{b\in\F_q^*}\mathcal{X}_b,
 \qquad
 \mathcal{X}_b=\{P(t):N_m(t)=b\}.
\]
Each $\mathcal{X}_b$ is a scattered $\Fq$-linear set of rank $m$, and
$L_{U_{\lambda_i}}\subseteq\mathcal{X}_{\lambda_i}$; see
\cite[Proposition 3.3]{DonatiDurante} and
\cite[Lemma 2.5 and Theorem 2.8]{ABCS}.

The line $AB$ is partitioned as
\[
 AB=\{A,B\}\ \dot\cup\!
       \bigcup_{a\in\F_q^*}\mathcal{J}_a,
\]
where
\[
 \mathcal{J}_a=
 \left\{
 (1,0,\ldots,0,(-1)^{\delta+1}u): N_m(u)=a
 \right\}.
\]
Every $\mathcal{J}_a$ is a scattered linear set of rank $m$ of pseudoregulus type, see \cite[Remark 2.2]{DonDur0} and \cite{Lun}. 

\begin{lem}\label{lem:roots-curve}
Let $f(X)=\sum_{j=0}^{\delta}h_jX^j$ be nonzero and $\lambda \in \F_{q}^*$, and let
$\mathcal H_f$ be the hyperplane of $\PG(\delta,q^m)$ defined by the equation
\[
        h_0X_0+\cdots+h_\delta X_\delta=0.
\]
Then:
\begin{enumerate}
\item $P(t)\in\mathcal H_f$ if and only if $f(t)=0$;
\item if $P(t),P(u)\in\mathcal{X}_b$ are distinct and both belong to
      $\mathcal H_f$, then $d_b(f)\geq 2$;
\item if $W\leq_{\F_q}U_{\lambda}$ and $\mathbb{P}(\langle W\rangle_{\F_{q^m}})\subseteq\mathcal H_f$, then
      $d_\lambda(f)\geq\dim_{\F_q}W$.
\end{enumerate}
\end{lem}

\begin{proof}
\textit{Point 1.} follows directly from the definition of right evaluation. Indeed, $P(t)=(1,t,N_2(t),\ldots,N_{\delta}(t))$ for some $t \in \F_{q^m}$ and $f(t)=\sum_{i=0}^{\delta}h_i N_i(t)$ (cf. \eqref{eva}). 

Choose
$a\in\F_{q^m}^*$ with $N_m(a)=b$. Any $t$ with norm $b$ has the
form $t=a\sigma(x)/x$ for some $x\in\F_{q^m}^*$.  The identity $N_j(t)=N_j(a)\sigma^j(x)/x$ 
shows that $f(t)=0$ if and only if $x\in\ker(f_a)$. Now, let $u=a \sigma(y)/y$.  Since $P(t), P(u)$ are distinct points of $\mathcal{X}_b$, then $x$ and $y$ are $\F_q$-independent, yielding \textit{Point 2}.

Finally, let $a \in \F_{q^m}^*$ such that $N_m(a)=\lambda$ and let $\boldsymbol{v}_a: \F_{q^m} \longrightarrow U_{\lambda}$ be an $\F_q$-linear isomorphism defined as
$$
\boldsymbol{v}_a(x)=\left( x, N_1(a)\sigma(x), \ldots, N_{\delta}(a)\sigma^{\delta}(x) \right), \qquad x \in \F_{q^m}.
$$

Let $H_f$ be the underlying vector space of $\mathcal{H}_f$. Then, $\varphi_a(x) \in H_f$ if, and only if, $h_0x+h_1N_1(a)\sigma(x)+\ldots+h_{\delta}N_{\delta}(a)\sigma^{\delta}(x)=0$. This is equivalent to $f_a(x)=0$, i.e., $x\in\ker f_a$. Thus, $\boldsymbol{v}_a(\ker f_a)=H_f\cap U_\lambda\supseteq W$ and $\dim_{\F_q} \boldsymbol{v}_a(\ker f_a)=\dim_{\F_q}\ker f_a\geq\dim_{\F_q}W$, proving \textit{Point 3}.

\end{proof}

\begin{lem}\label{lem:independent-spans}
Let $\mu_1,\ldots,\mu_s\in\F_q^*$ be pairwise distinct, choose
$c_k\in\F_{q^m}^*$ with $\operatorname{N}_{q^m/q}(c_k)=\mu_k$ and put
$$U_{\mu_k}=\{\boldsymbol{v}_{c_k}(x): x \in \F_{q^m}\}.$$ Let $W_k\leq_{\F_q}U_{\mu_k}$and put
$r=\sum_{k=1}^{s}\dim_{\F_q}W_k$. If $r\leq\delta+1$, then
\[
\dim_{\F_{q^m}}\langle W_1,\ldots,W_s\rangle_{\F_{q^m}}=r.
\]
\end{lem}
\begin{proof}
The assertion is immediate for $r=0$. Otherwise, choose an $\F_q$-basis
of each $W_k$, say $(\boldsymbol{v}_{c_k}(x_{k1}),\ldots,\boldsymbol{v}_{c_k}(x_{kr_k}))$,
where $r_k=\dim_{\F_q}W_k$. Form the $r\times r$ matrix $M$ whose columns are
the first $r$ coordinates of the vectors $\boldsymbol{v}_{c_k}(x_{kj})$.
If this matrix were singular, there exists a non-zero vector $\boldsymbol{h}=(h_0,h_1,\ldots,h_{r-1})$ such that $\boldsymbol{h}M=\bf{0}$. In particular, for any column of the matrix $M$ we have
$$
\sum_{i=0}^{r-1} h_i N_i(c_k)\sigma^i(x_{kj})=0,
$$
i.e. $f_{c_k}(x_{kj})=0$ with $f_{c_k}(x)=\sum_{i=0}^{r-1}h_i N_i(c_k) \sigma^i(x)$. Since $\boldsymbol{v}_{c_k}$ is an $\F_q$-linear isomorphism between $\F_{q^m}$ and $U_{\mu_k}$, we have that $x_{k1},\ldots,x_{kr_k}$ are $\F_q$-independent and thus $d_{\mu_k}(f)= \dim_{\F_q} \ker f_{c_k} \geq r_k$. Let $f(X)=\sum_{i=0}^{r-1}h_iX^i$, then by Proposition \ref{prop:rootspace} we get
\[
 r=\sum_{k=1}^{s}r_k\leq\sum_{k=1}^{s}d_{\mu_k}(f)
 \leq\deg(f)\leq r-1,
\]
a contradiction. Thus the chosen columns are independent over
$\F_{q^m}$, which proves the claim.
\end{proof}

\begin{remark}
We shall use the elementary fact that right multiplication by $X^i$ does
not change the root-space dimensions in nonzero norm classes:
\begin{equation}\label{eq:right-X}
        d_\lambda(gX^i)=d_\lambda(g),\qquad \lambda\in\F_q^*.
\end{equation}
This follows because the operator represented by $X^i$ is invertible.
\end{remark}

Let $\Lambda=\{\lambda_1,\ldots,\lambda_\ell\} \subseteq \F_q^*$ be a set of pairwise distinct elements. For some positive integer $r$, define
 \[
 \Pi_{r}(\Lambda)=
 \left\{
 \prod_{i=1}^{\ell}\lambda_i^{e_i}:
 0\leq e_i\leq n_i,\quad \sum_{i=1}^{\ell}e_i=r
 \right\}.
\]
Let $\F_q^*=S\ \dot\cup\ T$
be a partition of $\F_q^*$ such that
\begin{equation}\label{eq:switch-conditions} \tag{$\diamond$}
        \Lambda\subseteq T,
        \qquad
        S\Pi_{\delta-1}(\Lambda)\subseteq S.
\end{equation}
We define
\begin{equation}\label{eq:switched-set}
 \mathcal E_{S,T}=
 \{A,B\}\cup
 \bigcup_{b\in S}\mathcal{X}_b\cup
 \bigcup_{a\in T}\mathcal{J}_a.
\end{equation}
Since all components have $(q^m-1)/(q-1)$ points and are pairwise disjoint,
\begin{equation}\label{eq:size-E}
        |\mathcal E_{S,T}|=q^m+1.
\end{equation}

\begin{remark}\label{Scoset}
Let $G=\langle\Pi_{\delta-1}(\Lambda)\rangle$ be the multiplicative subgroup of $\F_q^*$. For any
$\pi\in\Pi_{\delta-1}(\Lambda)$, the inclusion  $S\pi\subseteq S$
is equivalent to $ S\pi=S$, because multiplication by $\pi$ is bijective. 
Consequently, the second condition in \eqref{eq:switch-conditions}
holds if and only if $S$ is a union of cosets of $G$.
Both conditions hold if and only if $S$ is a union of cosets of $G$
contained in $\F_q^*\setminus\Lambda G$ and $T=\F_q^*\setminus S$.
\end{remark}

\begin{proposition}\label{prop:admissible-partition}
Let $G=\langle\Pi_{\delta-1}(\Lambda)\rangle$. There exists a partition
$\F_q^*=S\,\dot\cup\,T$ with $S\neq\emptyset$ satisfying
\eqref{eq:switch-conditions} if and only if $\Lambda G\neq\F_q^*$.
\end{proposition}
\begin{proof}
If such a partition exists, Remark \ref{Scoset} gives
$\emptyset\neq S\subseteq\F_q^*\setminus\Lambda G$, so
$\Lambda G\neq\F_q^*$. Conversely, if $\Lambda G\neq\F_q^*$, consider
$T=\Lambda G$ and $S=\F_q^*\setminus T$. Then $S\neq\emptyset$,
$\Lambda\subseteq T$, and $S$ is a union of cosets of $G$.
\end{proof}

\begin{lem}\label{lem:switching-group}
We have
\[
 G=\left\langle\lambda_1^{\delta -1},
 \lambda_2/\lambda_1,\ldots,\lambda_\ell/\lambda_1\right\rangle,
 \qquad \Lambda G=\lambda_1G.
\]
In particular, a non-empty set $S$ satisfying Condition \eqref{eq:switch-conditions}
exists if and only if $G$ is a proper subgroup of $\F_q^*$.
If $1\in\Lambda$, then $G=\langle\Lambda\rangle$.
\end{lem}
\begin{proof}
For $\ell=1$, the formulas follow directly from the definition.
Suppose $\ell\geq2$. For any distinct indices $i,j  \in \{1,\ldots,\ell\}$, choose integers
$e_1,\ldots,e_\ell$ with sum $\delta -1$ such that  $e_i\leq n_i-1$, $e_j\geq1$ and
$0\leq e_k\leq n_k$ for any $k$ distinct from $i$ and $j$.
Such a choice is possible because $1\leq \delta -1\leq N-2$; under these
constraints the minimum and maximum possible sums are $1$ and $N-1$.
If $e=(e_1,\ldots,e_{\ell})$ is an admissible exponent vector in the definition of $\Pi_{\delta-1}(\Lambda)$, then $e+\varepsilon_i-\varepsilon_j$,  where $\varepsilon_i$ denotes the $i$-th unit vector, is admissible as well.

Taking the ratio of the corresponding products gives
$\lambda_i/\lambda_j\in G$.
Every admissible product has the form
\[
 \prod_{i=1}^{\ell}\lambda_i^{e_i}
 =\lambda_1^{\delta-1}\prod_{i=2}^{\ell}(\lambda_i/\lambda_1)^{e_i}.
\]
Since an admissible exponent vector exists, this identity also gives
$\lambda_1^{\delta -1}\in G$. It proves the asserted description of $G$.
All elements of $\Lambda$ lie in the same coset $\lambda_1G$, so
$\Lambda G=\lambda_1G$. The existence assertion follows from
Proposition \ref{prop:admissible-partition}. If $1\in\Lambda$, putting $\lambda_1=1$ we obtain $G=\langle\Lambda\rangle$.
\end{proof}

\begin{lem}
\label{lem:points-exterior}
 Every point of $\mathcal E_{S,T}$ lies outside
$\Omega_{\delta-2}(\mathcal{U}_\Lambda)$.
\end{lem}
\begin{proof}
    Suppose that $P\in \mathcal{E}_{S,T}$ belongs to
$\Omega_{\delta-2}(\mathcal{U}_\Lambda)$.  Then $P$ lies in $\mathcal{W}=\mathbb{P}(W)$ where $W=\langle W_1,\ldots,W_\ell \rangle_{\F_{q^m}}$, with $W_i\leq_{\F_q}U_{\lambda_i}$ and  $r_i=\dim_{\F_q} W_i$. We may assume $\sum_i r_i=\delta-1$, and thus by Lemma \ref{lem:independent-spans} the subspace $W$ has dimension $\delta -1$. A hyperplane containing $\mathcal{W}$ can be seen as a linear operator which vanishes on $W$. Let 
$$W^*=\left\{ h(x)=\sum_{i=0}^{\delta} h_i x_i : h(x)=0 \,\, \forall \,\, x \in W \right\}.
$$
We have that $\dim_{\F_{q^m}} W^* = 2$. Consider the following $\delta$-dimensional $\F_{q^m}$-subspace of linear operators
$$
T_{j}=\left \{ h(x)=\sum_{i=0}^{\delta} h_i x_i: h_{j}=0 \right\}.
$$
Then, we have
$$
\dim_{\F_{q^m}}(W^* \cap T_{j}) \geq \dim_{\F_{q^m}} W^* + \dim_{\F_{q^m}} T_{j} - (\delta +1) = 1.
$$
Therefore, for any $j \in \{0,1,\ldots,\delta\}$, there exists a polynomial $f_j=\sum_{\overset{i=0}{i \neq j}}^{\delta}h_i X^i$ such that the hyperplane $\mathcal{H}_{f_j}$ containing $\mathcal{W}$ having equation
$$
\sum_{\overset{i=0}{i \neq j}}^\delta h_i X_i=0.
$$
We will distinguish three cases.
\begin{itemize}
    \item $P=P(t) \in \mathcal{X}_b$ with $b \in S$.\\ Let $\mathcal{H}_{f_{\delta}}:h_0 X_0 + \ldots + h_{\delta -1}X_{\delta }=0$ be a hyperplane containing $\mathcal{W}$. By \textit{Point 3} of Lemma \ref{lem:roots-curve} and Proposition \ref{prop:rootspace}, we have
    $$
    \deg(f_{\delta}) \geq \sum_{\lambda \in \Lambda \cup \{ b\}}d_{\lambda}(f_{\delta}) \geq \sum_{i=1}^{\ell} \dim_{\F_q}W_i + d_b(f_{\delta}).
    $$
    Since $P \in \mathcal{H}_{f_{\delta}} \cap \Omega_{\delta-2}(\mathcal{U}_\Lambda)$, then $d_b(f_{\delta}) \geq 1$. Therefore, $\deg(f_{\delta}) \geq \delta$, a contradiction follows.
    \item $P \in \mathcal{J}_a$ with $a \in T$.\\
    Let $\mathcal{H}_{f_{\delta}}:h_0 X_0 + \ldots + h_{\delta -1}X_{\delta -1}=0$ be a hyperplane containing $\mathcal{w}$. Since $P \in \mathcal{H}_{f_\delta}$, then $h_0=0$. Now,
    $$
    f_\delta=\sum_{i=1}^{\delta -1} h_i X^i =\left( \sum_{i=0}^{\delta -2}h_{i+1} X^i \right) X =:g X,
    $$
    where $\deg(g) \leq \delta -2$. By \eqref{eq:right-X}, $d_\lambda(f_\delta)=d_\lambda(g)$ for every nonzero norm class. Therefore,
    $$
    \deg(g) \geq \sum_{i=1}^{\ell}d_{\lambda_i}(g)=\sum_{i=1}^{\ell}d_{\lambda_i}(f_\delta)\geq \sum_{i=1}^{\ell} \dim_{\F_q}W_i=\delta -1,
    $$
    again a contradiction.
    \item $P \in \{A,B\}$. \\ 
    Assume $P=A$ and let $\mathcal{H}_{f_0}: h_1X_1 + \ldots + h_{\delta}X_{\delta}=0$  be a hyperplane containing $\mathcal{W}$.  Since $A\in \mathcal{W}$, every element of $W^*$ has $h_\delta=0$. Thus $f_0=hX$ for some polynomial $h$ with $\deg(h)\leq\delta-2$, and the previous root-space argument gives a contradiction.
    Assume $P=B$ and let $\mathcal{H}_{f_\delta}$ be a hyperplane containing $\mathcal{W}$. Since $B\in \mathcal{H}_{f_\delta}$ then $h_{0}=0$. Therefore, it follows a contradiction repeating the same argument of the previous case. 
\end{itemize}
Hence, the claim follows.
\end{proof}

\begin{thm}
\label{thm:main-exterior}
The set $\mathcal E_{S,T}$ is an exterior set with respect to
$\Omega_{\delta-2}(\mathcal{U}_\Lambda)$ if and only if the Condition \eqref{eq:switch-conditions} holds.
\end{thm}

\begin{proof}
    First assume Condition  \eqref{eq:switch-conditions}. Suppose that there are two distinct points $P,Q \in \mathcal{E}_{S,T}$ such that the line $PQ$ meets non-trivially the set $\Omega_{\delta-2}(\mathcal U_\Lambda)$. We may assume that there exist subspaces $W_i \leq_{\F_q} U_{\lambda_i}$ with $r_i=\dim_{\F_q}W_i$ and $\sum_i r_i = \delta -1$ such that
    $$
    PQ \cap  \mathcal{W}  \ne \emptyset.
    $$
    where $\mathcal{W}=\mathbb{P}(\langle W_1, \ldots, W_{\ell} \rangle_{\F_{q^m}}) $ and by Lemma \ref{lem:independent-spans}, it has projective dimension $\delta -2$. By Lemma \ref{lem:points-exterior}, the line $PQ$ meets $\mathcal{W}$ in a point $R$ different from $P$ and $Q$. The projective span $\langle P,\mathcal W\rangle$ contains $PQ$ and it has  projective dimension $\delta-1$. Let $ \langle P , \mathcal{W} \rangle=\mathcal H_f:h_0X_0+\cdots+h_\delta X_\delta=0$ be this hyperplane with $f=\sum_{i=0}^{\delta}h_i X^i$. By \textit{Point 3} of Lemma \ref{lem:roots-curve}, we have $d_{\lambda_i}(f)\geq r_i$ for any $i=1,\ldots,\ell$. We distinguish four cases.
    \begin{itemize}
        \item $P \in \mathcal{X}_b, Q \in \mathcal{X}_c$ with $b,c \in S$. 

        If $b \ne c$,  then $d_b(f)\geq 1$, $d_c(f)\geq 1$  and 
        $$
        \deg(f) \geq \sum_{\lambda \in \Lambda \cup \{b,c\}} d_{\lambda}(f) \geq \sum_{i=1}^{\ell} \dim_{\F_q}W_i + d_b(f)+d_c(f)=\delta +1.
        $$

        If $b=c$, then by \textit{Point 2} of Lemma \ref{lem:roots-curve} we have $d_b(f) \geq 2$ and 
        $$
        \deg(f) \geq \sum_{\lambda \in \Lambda \cup \{b\}} d_{\lambda}(f) \geq \delta +1.
        $$

        In both cases, a contradiction follows.
        \item $P,Q \in AB$. 

        In this case the line $AB$ is contained in the hyperplane $\mathcal{H}_f$ and so $h_0=h_{\delta}=0$. Applying the same argument of the proof of Lemma \ref{lem:points-exterior}, we get $f=g X$ with $\deg(g)\leq \delta -2$. Therefore, we have
        $$
        \deg(g) \geq \sum_{\lambda \in \Lambda} d_{\lambda}(g)= \sum_{\lambda \in \Lambda}d_{\lambda}(f) \geq \delta -1,
        $$
        yielding a contradiction.
        \item $P \in \mathcal{X}_b, Q \in \mathcal{J}_a$ with $b \in S$ and $a \in T$.  

Assume $Q=(1,0,\ldots,0,(-1)^{\delta +1}u)$ with $N_m(u)=a$. Since $Q \in \mathcal{H}_f$, then $h_0+(-1)^{\delta +1}u h_{\delta}=0$. If $h_{\delta}=0$, then $h_0=0$, so that $f=gX$ with $\deg(g)\leq\delta-2$; by \eqref{eq:right-X}, Point 1 of Lemma \ref{lem:roots-curve} applied to $P$, and Proposition \ref{prop:rootspace} with $\Xi=\Lambda\cup\{b\}$, we get
\begin{align*}
\delta-2\ \geq\ \deg(g)&\geq\sum_{\lambda\in\Lambda\cup\{b\}}d_\lambda(g)
=\sum_{\lambda\in\Lambda\cup\{b\}}d_\lambda(f)\\
&\geq\sum_{i=1}^{\ell}\dim_{\F_q}W_i+d_b(f)\ \geq\ (\delta-1)+1,
\end{align*}
a contradiction. Hence, it holds $h_{\delta}h_0 \ne 0$. Now, we have
$$
\delta = \deg(f) \geq \sum_{\lambda \in \Lambda \cup \{b \}}d_{\lambda}(f)\geq \sum_{i=1}^{\ell}\dim_{\F_q}W_i + d_b(f)\geq (\delta -1) +1= \delta.
$$
So, $$\deg(f)=\sum_{\lambda \in \Lambda \cup \{b\}}d_{\lambda}(f).$$

Applying Proposition \ref{prop:viete} with $\Xi=\Lambda\cup\{b\}$, we obtain
$$
N_m\left(\frac{h_0}{h_{\delta}}\right)=(-1)^{\delta m}b \prod_{i=1}^{\ell}\lambda_i^{r_i}.
$$
Recall that $u=(-1)^{\delta}\frac{h_0}{h_{\delta}}$, then we have
\begin{align*}
    a=N_m(u)&=N_m\left( (-1)^{\delta}\right) N_m\left( \frac{h_0}{h_{\delta}} \right) \\ &= (-1)^{\delta m} (-1)^{\delta m} b \prod_{i=1}^{\ell} \lambda_i^{r_i}\\ &= b \prod_{i=1}^{\ell} \lambda_i^{r_i} \in S \Pi_{\delta -1}(\Lambda) \subseteq S,
\end{align*}
a contradiction follows.
        
        \item $P \in \{A,B\}, Q \in \mathcal{X}_b$ with $b \in S$. 

        If $P=A$, then $h_\delta=0$, so
$\deg(f)\leq\delta-1$. We have
$$
\deg(f) \geq \sum_{\lambda \in \Lambda \cup \{b\}}d_{\lambda}(f) \geq \sum_{i=1}^{\ell}\dim_{\F_q}W_i + d_b(f) \geq \delta.
$$ 

If $P=B$, then $h_0=0$ and
$f=gX$ with $\deg(g)\leq\delta-1$.  We have
$$
\deg(g)\geq \sum_{\lambda \in \Lambda \cup \{b\}} d_{\lambda}(g)=\sum_{\lambda \in \Lambda \cup \{b\}} d_{\lambda}(f) \geq \delta.
$$

In both cases, a contradiction follows.

    \end{itemize}

Conversely, suppose that $\mathcal E_{S,T}$ is an exterior set with respect to $\Omega_{\delta-2}(\mathcal U_\Lambda)$. We first prove that $\Lambda\subseteq T$. Suppose, on the contrary, that $\lambda\in\Lambda\cap S$. Then \[ \mathcal X_\lambda \subseteq\Omega_0(\mathcal U_\Lambda) \subseteq\Omega_{\delta-2}(\mathcal U_\Lambda). \] Since $\mathcal X_\lambda\subseteq\mathcal E_{S,T}$, then a contradiction follows. Therefore, $\Lambda\subseteq T.$ We now prove that $S\Pi_{\delta-1}(\Lambda)\subseteq S$. Suppose, on the contrary, that this inclusion does not hold. Then there exist $b\in S$ and integers \[ 0\leq r_i\leq n_i, \qquad \sum_{i=1}^{\ell}r_i=\delta-1, \] such that \[ a=b\prod_{i=1}^{\ell}\lambda_i^{r_i}\in T. \] 
 For every $i$, choose an $\F_q$-subspace \[ W_i\leq_{\F_q}U_{\lambda_i}, \qquad \dim_{\F_q}W_i=r_i, \] and put \[ V=\langle W_1,\ldots,W_\ell\rangle_{\F_{q^m}}, \qquad \mathcal W=\mathbb P(V). \] By Lemma~\ref{lem:independent-spans}, \[ \dim_{\F_{q^m}}V=\sum_{i=1}^{\ell}r_i=\delta-1, \] and hence $\dim\mathcal W=\delta-2$. In particular, $\mathcal W\subseteq\Omega_{\delta-2}(\mathcal U_\Lambda).$ Choose a point $P=P(t)\in\mathcal X_b$. Let $c\in\F_{q^m}^*$ satisfy $N_m(c)=b$. Then there exists $x\in\F_{q^m}^*$ such that $t=c \sigma(x)/x$, and hence $P=\langle \boldsymbol{v}_c(x)\rangle$. Consider the $\F_q$-subspace $W_b=\langle \boldsymbol{v}_c(x)\rangle_{\F_q}$. Since the norm classes in $\Lambda\cup\{b\}$ are pairwise distinct, Lemma~\ref{lem:independent-spans}, applied to $W_1,\ldots,W_\ell,W_b$, gives \[ \dim_{\F_{q^m}} \langle W_1,\ldots,W_\ell,W_b\rangle_{\F_{q^m}} =\delta. \] Therefore, $\mathcal H_f=\langle\mathcal W,P\rangle$ is a hyperplane of $\PG(\delta,q^m)$. Write its equation as \[ h_0X_0+\cdots+h_\delta X_\delta=0 \] and set \[ f(X)=\sum_{j=0}^{\delta}h_jX^j. \] Since $\mathbb P(\langle W_i \rangle_{\F_{q^m} })\subseteq\mathcal H_f$, Point~\textit{3} of Lemma~\ref{lem:roots-curve} gives $d_{\lambda_i}(f)\geq r_i$ for any $i \in \{1,\ldots,\ell\}$. Moreover, since $P=P(t)\in\mathcal H_f$, Point~\textit{1} of Lemma~\ref{lem:roots-curve} gives $f(t)=0$, and hence $d_b(f)\geq1$. Thus, by Proposition~\ref{prop:rootspace}, \[ \delta =\sum_{i=1}^{\ell}r_i+1 \leq \sum_{i=1}^{\ell}d_{\lambda_i}(f)+d_b(f) \leq \deg(f) \leq\delta. \] Consequently, equality holds throughout. In particular, \[ \deg(f)=\delta, \qquad d_{\lambda_i}(f)=r_i \quad\text{for every }i, \qquad d_b(f)=1. \] Since $\deg(f)=\delta$, we have $h_\delta\neq0$. We also claim that $h_0\neq0$. Otherwise, $f=gX$ for some skew polynomial $g$ with $\deg(g)\leq\delta-1$. By \eqref{eq:right-X}, $d_\lambda(g)=d_\lambda(f)$ for every nonzero norm class $\lambda$. Hence \[ \delta = \sum_{i=1}^{\ell}d_{\lambda_i}(f)+d_b(f) = \sum_{i=1}^{\ell}d_{\lambda_i}(g)+d_b(g) \leq\deg(g) \leq\delta-1, \] a contradiction. Therefore $h_0h_\delta\neq0$. Applying Proposition~\ref{prop:viete} to the set $\Lambda\cup\{b\}$, we obtain \[ N_m\left(\frac{h_0}{h_\delta}\right) = (-1)^{\delta m} b\prod_{i=1}^{\ell}\lambda_i^{r_i}. \] Put $u=(-1)^\delta\frac{h_0}{h_\delta}$.  Then \[ N_m(u) = N_m((-1)^\delta) N_m\left(\frac{h_0}{h_\delta}\right) = b\prod_{i=1}^{\ell}\lambda_i^{r_i} =a. \] Therefore the point \[ Q=(1,0,\ldots,0,(-1)^{\delta+1}u) \] belongs to $\mathcal J_a\subseteq\mathcal E_{S,T}$. Moreover, \[ h_0+(-1)^{\delta+1}uh_\delta = h_0+(-1)^{\delta+1}(-1)^\delta h_0 =0, \] so $Q\in\mathcal H_f$. Since $\delta \geq 2$, the points $P$ and $Q$ are distinct: indeed, the second coordinate of $P=P(t)$ is $t\neq0$, whereas the second coordinate of $Q$ is zero. Finally, because $\mathcal E_{S,T}$ is exterior, each of its points lies outside $\Omega_{\delta-2}(\mathcal U_\Lambda)$. Since $\mathcal W\subseteq\Omega_{\delta-2}(\mathcal U_\Lambda)$, we have $P,Q\notin\mathcal W$. On the other hand, $\mathcal W$ has projective dimension $\delta-2$ and is contained in the hyperplane $\mathcal H_f$, which has projective dimension $\delta-1$. Thus $\mathcal W$ is a hyperplane of $\mathcal H_f$. Since $PQ\subseteq\mathcal H_f$, the line $PQ$ meets $\mathcal W$. Therefore, \[ PQ\cap\Omega_{\delta-2}(\mathcal U_\Lambda)\neq\emptyset, \] contradicting the hypothesis on $\mathcal E_{S,T}$. Hence $S\Pi_{\delta-1}(\Lambda)\subseteq S$, and the proof is complete.
    \end{proof}

\section{The family of MSRD codes}\label{thefamily:cst}

Throughout this section, assume that the partition
$\F_q^*=S\,\dot\cup\,T$ satisfies Condition \eqref{eq:switch-conditions}. Then, consider the sets
    $$
    \hat{\mathcal{E}}=\{ \alpha \bold{v} : \alpha \in \F_{q^m}, \langle \bold{v} \rangle \in \mathcal{E}_{S,T} \} \subseteq \F_{q^m}^{\delta +1}
    $$
and $\mathcal{C}_{S,T}=\syn_{\bold{H}}^{-1}(\hat{\mathcal{E}}) \subseteq \F_{q^m}^N$. In the following, given two disjoint subsets $\mathcal{A}$ and $\mathcal{B}$ we will denote by $\mathbb{K}(\mathcal{A},\mathcal{B})$ the cone having vertex $\mathcal{A}$ and base $\mathcal{B}$.

\begin{remark}\label{projcone}
    Let $\mathcal{A},\mathcal{B}$ two non-empty subsets of a vector space $\F_{q^m}^{t+1}$. Assume that $\mathcal{A},\mathcal{B}$ are both closed under $\F_{q^m}$-multiplication, $\mathcal{A} \cap \mathcal{B} = \{\bf{0}\}$ and define the subset 
    $$
    \mathcal{M}=\mathcal{A}+\mathcal{B}=\{ \bold{a}+\bold{b} : \bold{a} \in \mathcal{A}, \bold{b} \in \mathcal{B} \}.
    $$
    We claim that $\mathbb{P}(\mathcal{M})=\mathbb{K}(\mathbb{P}(\mathcal{A}),\mathbb{P}(\mathcal{B})) \subseteq \PG(t,q^m)$. Let $\bold{z}$ be a non-zero vector of $\mathcal{M}$, then there are $\bold{a} \in \mathcal{A}, \bold{b} \in \mathcal{B}$ such that $\bold{z}=\bold{a}+\bold{b}$. Hence, the point $\langle \bold{z} \rangle$ belongs to the line through $\langle \bold{a} \rangle$ and $\langle \bold{b} \rangle$. This yields $\mathbb{P}(\mathcal{M}) \subseteq \mathbb{K}(\mathbb{P}(\mathcal{A}),\mathbb{P}(\mathcal{B}))$. Now, let $\langle \bold{c} \rangle \in \mathbb{K}(\mathbb{P}(\mathcal{A}),\mathbb{P}(\mathcal{B}))$. Then, there are non-zero vectors $\bold{u} \in \mathcal{A},\bold{v} \in \mathcal{B}$ such that $\langle \bold{c} \rangle$ belongs to the line through $\langle \bold{u} \rangle$ and $\langle \bold{v} \rangle$. This yields there exist $\lambda, \mu \in \F_{q^m}$ such that
    $$
    \bold{c}=\lambda \bold{u} + \mu \bold{v},
    $$
    and since $\mathcal{A},\mathcal{B}$ are closed under $\F_{q^m}$-multiplication then $\lambda \bold{u} \in \mathcal{A}$ and $\mu \bold{v} \in \mathcal{B}$. Hence, $\langle \bold{c} \rangle \in \mathbb{P}(\mathcal{M})$.
\end{remark}

\begin{proposition}\label{conecode}
    Let $K$ be the kernel of the syndrome map $\syn_\bold{H}$ and let $L$ be an $\F_{q^m}$-subspace of $\F_{q^m}^N$ such that $\F_{q^m}^N=K \oplus L$. If
    $$
    \tilde{\mathcal{E}}=\mathbb{P}\left( (\syn_\bold{H} \vert_L)^{-1}(\hat{\mathcal{E}}) \right) \subseteq \mathbb P(L),
    $$
    then $\tilde{\mathcal{E}}$ is projectively equivalent to $\mathcal{E}_{S,T}$ and 
    $$
    \mathbb{P}(\mathcal{C}_{S,T})=\mathbb{K}(\mathbb P(K),\tilde{\mathcal{E}}),
    $$
    or equivalently
    $$
    \mathcal{C}_{S,T}=K + (\syn_\bold{H} \vert_L)^{-1}(\hat{\mathcal{E}}).
    $$
\end{proposition}

\begin{proof}
    Since the matrix $\bold{H}$ has full-rank, then $\dim_{\F_{q^m}}L=\delta +1$, i.e. $\syn_\bold{H} \vert_L$ is an isomorphism. In particular, $\tilde{\mathcal E}$ is projectively equivalent to $\mathcal E_{S,T}$ via the projectivity induced by $\syn_{\bold{H}}\big|_L$. We now prove the equality
$$
\mathcal C_{S,T}
=
K+
\left(\syn_\bold{H}\vert_L\right)^{-1}
(\hat{\mathcal E}).
$$
Put
$$
\mathcal{E}'
=
\left(\syn_\bold{H}\vert_L\right)^{-1}
(\hat{\mathcal E})
\subseteq L.
$$
Note that $\bf{0}\in\hat{\mathcal E}$, and hence $\bf{0}\in\mathcal{E}'$. Let $\bold{z}\in\mathcal C_{S,T}$. By definition,
$\syn_\bold{H}(\bold{z})\in\hat{\mathcal E}$.
Since $\F_{q^m}^{N}=K\oplus L$,
there exist unique vectors $\bold{k}\in K$ and $\bold{x}\in L$ such that $\bold{z}=\bold{k}+\bold{x}$. So, we get $\syn_{\bold{H}}(\bold{z})=\syn_\bold{H}(\bold{x}) \in \hat{\mathcal{E}}$. Moreover, $\bold{x} \in L$ and hence $\bold{x} \in \left(\syn_\bold{H}\vert_L\right)^{-1}
(\hat{\mathcal E})$. It yields 
$$
\bold{z}=\bold{k}+\bold{x} \in K + \mathcal{E}',
$$
i.e. $\mathcal{C}_{S,T} \subseteq K + (\syn_\bold{H} \vert_L)^{-1}(\hat{\mathcal{E}})$. Conversely, let $\bold{z}\in
K+
\left(\syn_\bold{H}\vert_L\right)^{-1}
(\hat{\mathcal E})$.
Then there exist $\bold{k}\in K$ and $\bold{x}\in
\left(\syn_\bold{H}\vert_L\right)^{-1}
(\hat{\mathcal E})$
such that $\bold{z}=\bold{k}+\bold{x}$. So, we have
$$
\syn_\bold{H}(\bold{z})=\syn_\bold{H}(\bold{x})=\syn_\bold{H} \vert_L(\bold{x}) \in \hat{\mathcal{E}}.
$$
By definition of $\mathcal C_{S,T}$, we get $\bold{z}\in \syn_\bold{H}^{-1}(\hat{\mathcal E})
=
\mathcal C_{S,T}$. Hence $K+
\left(\operatorname{syn}_\bold{H}\vert_L\right)^{-1}
(\hat{\mathcal E})
\subseteq
\mathcal C_{S,T}$. Finally, the equality
$$
\mathbb{P}(\mathcal{C}_{S,T})=\mathbb{K}(\mathbb P(K),\tilde{\mathcal{E}})
$$
follows from Remark \ref{projcone}. The claim follows.
\end{proof}

\begin{lem}\label{lem:kernel-distance} Let $\bold{H}$ be the full-rank block matrix in \eqref{H} and put  $K=\ker(\syn_\bold{H})$. If $\delta\leq N-2$, then $K$ is an $\F_{q^m}$-linear MSRD code of dimension $N-\delta-1$ and \[ d_{\rm SR}(K)=\delta+2. \] If $\delta=N-1$, then $K=\{\boldsymbol{0}\}$. \end{lem}

\begin{proof} Let $\mathcal L=\operatorname{rowsp}_{\F_{q^m}}(\bold{H})$.  By construction, $\mathcal L$ is a linearized Reed--Solomon code of length $N$ and dimension $\delta+1$. In particular, it is an $\F_{q^m}$-linear MSRD code. Since $\bold{H}$ is a generator matrix of $\mathcal L$, we have \[ K = \{\bold{z}\in\F_{q^m}^N:\bold{H}\bold{z}^{\mathsf T}= \boldsymbol{0}\} = \mathcal L^{\perp\!\!\!\!\perp}. \] 

 The vector dual of a linearized Reed--Solomon code is again a linearized Reed--Solomon code of complementary dimension, see \cite[Theorem 4]{MartinezKschischang}. Therefore, if $\delta\leq N-2$, then \[ \dim_{\F_{q^m}}K = N-(\delta+1) = N-\delta-1 \] and \[ d_{\rm SR}(K) = N-\dim_{\F_{q^m}}K+1 = \delta+2. \] If $\delta=N-1$, then $\dim_{\F_{q^m}}K=0$, and hence $K=\{\boldsymbol{0}\}$. \end{proof}

We note that $|\hat{\mathcal{E}}|=(q^m+1)(q^m-1)+1=q^{2m}$. 

\begin{thm}\label{thm:MSRD}
    The code $\mathcal{C}_{S,T}$ satisfies
    $$
    |\mathcal{C}_{S,T}|=q^{m(N-\delta +1)} \qquad \mathrm{and} \qquad d_{SR}(\mathcal{C}_{S,T})=\delta,
    $$
    i.e. $\mathcal{C}_{S,T}$ is a MSRD code. Moreover, it is closed under $\F_{q^m}$-multiplication.
\end{thm}

\begin{proof}
    By Proposition \ref{conecode}  and since the matrix $\bold{H}$ has rank $\delta +1$, we get
    $$
    |\mathcal{C}_{S,T}|=|K| \cdot |\hat{\mathcal{E}}|=q^{m(N-\delta -1)} \cdot q^{2m}=q^{m(N-\delta +1)}.
    $$
    Now, let $\bold{z},\bold{z}'$ be two distinct codewords of $\mathcal{C}_{S,T}$ and put $\bold{c}=\bold{z}-\bold{z}'$.

    Suppose $\syn_\bold{H}(\bold{c})=0$. Since $\bold{c}\neq0$, Lemma~\ref{lem:kernel-distance}
gives $\operatorname{wt}_{\rm SR}(\bold{c})\geq\delta+2$. This case cannot occur
when $\delta=N-1$, because then $K=\{\boldsymbol{0}\}$.

    Suppose $\syn_\bold{H}(\bold{c}) \ne \boldsymbol{0}$. Since $\syn_\bold{H}(\bold{z}), \syn_\bold{H}(\bold{z}') \in \hat{\mathcal{E}}$, then the projective point $\langle \syn_\bold{H}(\bold{c}) \rangle$ either belongs to a secant line through two distinct points of $\mathcal{E}_{S,T}$ or is a point of $\mathcal{E}_{S,T}$. In both cases, by Lemma \ref{lem:points-exterior} and  Theorem \ref{thm:main-exterior}, we have $\langle \syn_\bold{H}(\bold{c}) \rangle \not \in \Omega_{\delta -2}(\mathcal{U}_\Lambda)$ and so by Lemma \ref{lem:syndrome} it holds $\mathrm{wt_{\mathrm{SR}}}(\bold{c}) \geq \delta$.

We have proved $d_{\rm SR}(\mathcal C_{S,T})\geq\delta$.
Its cardinality and the Singleton bound imply that $d_{\mathrm{SR}}(\cC_{S,T})=\delta$. Finally, $\syn_\bold{H}(\alpha \bold{z})=\alpha\syn_\bold{H}(\bold{z})$ and
$\hat{\mathcal E}$ is closed under multiplication by every
$\alpha\in\F_{q^m}$; therefore the code has the same closure property.
    \end{proof}

    \begin{remark} There are two trivial cases. \begin{itemize}
    \item Suppose that $S=\emptyset$. Then $T=\F_q^*$ and \[ \mathcal E_{\emptyset,\F_q^*} = \{A,B\}\cup\bigcup_{a\in\F_q^*}\mathcal{J}_a = AB. \] Consequently, $\hat{\mathcal E} = \langle \bold{e}_0, \bold{e}_\delta \rangle_{\F_{q^m}}$, (cf. \eqref{puntiAB}) and hence \[ \mathcal C_{\emptyset,\F_q^*} = \syn_\bold{H}^{-1} \bigl(\langle \bold{e}_0,\bold{e}_\delta\rangle_{\F_{q^m}}\bigr) \] is an $\F_{q^m}$-linear code. More explicitly, let $\bold{H}'$ be the $(\delta-1)\times N$ submatrix of $\bold{H}$ formed by the rows with indices $2,\ldots,\delta$. Then \[ \mathcal C_{\emptyset,\F_q^*} = \ker \bold{H}'= (\operatorname{rowsp}_{\F_{q^m}}(\bold{H}'))^{\perp\!\!\!\!\perp}. \] For every $i$, define \[ \mathcal B_i' = (\beta_{i1}',\ldots,\beta_{im}'), \qquad \beta_{ih}'=a_i\sigma(\beta_{ih}), \] and put \[ \mathfrak B' = (\mathcal B_1',\ldots,\mathcal B_\ell'). \] Each $\mathcal B_i'$ is an $\F_q$-basis of $\F_{q^m}$, since $x\mapsto a_i\sigma(x)$ is invertible. Furthermore, \[ N_r(a_i)\sigma^r(\beta_{ih}') = N_{r+1}(a_i)\sigma^{r+1}(\beta_{ih}), \qquad 0\leq r\leq\delta-2. \] It follows that \[ \operatorname{rowsp}_{\F_{q^m}}(\bold{H}') = \operatorname{ev}_{\boldsymbol a,\boldsymbol{\mathcal B}'} (\mathrm{LRS}_{\delta-1}). \] Hence $\mathcal C_{\emptyset,\F_q^*}$ is the vector dual of an LRS code of dimension $\delta-1$. By \cite[Theorem 4]{MartinezKschischang}, it is therefore an LRS code of complementary dimension $N-\delta+1$ and, in particular, an $\F_{q^m}$-linear MSRD code with minimum sum-rank distance $\delta$. 
    \item If $q=2$, then $\ell=1$ and $\Lambda=\{1\}$. The condition $\Lambda\subseteq T$ then forces \[ T=\F_q^*, \qquad S=\emptyset, \] so the construction reduces to the preceding linear case. \end{itemize} \end{remark}

\begin{remark}\label{rem:nonadditive}
If $S\neq\emptyset$, then $\mathcal C_{S,T}$ is non-additive, for every
$\ell \geq 1$. Indeed, additivity together with closure under all
$\F_{q^m}$-scalars would make the code an $\F_{q^m}$-subspace. Its image
$\hat{\mathcal E}$ under the syndrome map would then be a subspace of
dimension $2$, since $|\hat{\mathcal E}|=q^{2m}$. Since it contains $\bold{e}_0$ and
$\bold{e}_{\delta}$, then $\mathcal{E}_{S,T}$ must be the line $AB$. This is a contradiction because every point $P(t)\in\mathcal X_b$, with $b \in S$, has a nonzero second
coordinate and lies outside $AB$.
\end{remark}

As byproduct of remark above, we have the following.

\begin{cor}\label{cor:no-line}
If $S\neq\emptyset$, then $\mathcal E_{S,T}$ contains no line of
$\PG(\delta,q^m)$.
\end{cor}

\subsection{Maximal linear subcodes, a geometric invariant and weight distribution}

Let $\bold{H}\in\F_{q^m}^{(\delta+1)\times N}$ be any
matrix of full rank, consider the associated syndrome map $\syn_{\bold{ H}}$ and  $K=\ker(\syn_{\bold H})$. For each point
$P=\langle\bold{p}\rangle\in\PG(\delta,q^m)$, define
\begin{equation}\label{DP}
\mathcal D_P
=
\syn_{\bold H}^{-1}
\bigl(\langle\bold{p}\rangle_{\F_{q^m}}\bigr).
\end{equation}
Since $\syn_{\bold{H}}$ is surjective, each $\mathcal D_P$
is an $\F_{q^m}$-linear subspace containing $K$, with $\dim_{\F_{q^m}}\mathcal D_P= N -\delta$.
Moreover,  if $P \neq Q$,
$\mathcal D_P\cap\mathcal D_Q=K$.
For the particular syndrome matrix in \eqref{H}, these subspaces
play a role in the sum-rank setting analogous to that of the
generalized Gabidulin subcodes used in \cite{DGL} to distinguish
$\mathcal C_{\sigma,T}$ from the Otal--\"Ozbudak family.

\begin{remark}\label{rm:dec}
Let $\F_q^*=S\,\dot\cup\,T$ be a partition satisfying Condition~\eqref{eq:switch-conditions}. Note that \[ \hat{\mathcal E}\setminus\{\boldsymbol{0}\} = \dot\bigcup_{\langle \bold{p} \rangle \in\mathcal{E}_{S,T}} \bigl(\langle \bold{p}\rangle_{\F_{q^m}}\setminus\{\bold{0}\}\bigr). \] Taking preimages under $\syn_\bold{H}$ gives \begin{equation*}\label{eq:decomposition} \mathcal C_{S,T} = K\ \dot\cup\!\! \bigcup_{P\in\mathcal E_{S,T}} \bigl(\mathcal D_P\setminus K\bigr). \end{equation*}
\end{remark}

\begin{thm}\label{thm:maximal-subcodes}
Let $\F_q^*=S\,\dot\cup\,T$ be a partition satisfying
Condition~\eqref{eq:switch-conditions}, with $S\neq\emptyset$,
and let $\bold{H}$ be the syndrome matrix in as in \eqref{H}.
Then $\mathcal C_{S,T}$ contains exactly $q^m+1$ maximal
$\F_{q^m}$-linear subcodes. These are precisely the $(N-\delta)$-spaces $\mathcal D_P$, with $P\in\mathcal E_{S,T}$.
\end{thm}

\begin{proof} 
Let $\mathcal D$ be an $\F_{q^m}$-linear subspace contained in $\mathcal C_{S,T}$. Since $\syn_\bold{H}$ is $\F_{q^m}$-linear, $\syn_\bold{H}(\mathcal D)$ is an $\F_{q^m}$-subspace contained in $\hat{\mathcal E}$. We claim that \[ \dim_{\F_{q^m}}\syn_\bold{H}(\mathcal D)\leq1. \] Indeed, if this dimension were at least $2$, then $\mathbb P(\syn_\bold{H}(\mathcal D))$ would contain a projective line. Since \[ \mathbb P(\syn_\bold{H}(\mathcal D)) \subseteq \mathbb P(\hat{\mathcal E}) = \mathcal E_{S,T}, \] this would contradict Corollary~\ref{cor:no-line}. Thus either $\syn_\bold{H}(\mathcal D)=\{\boldsymbol{0}\}$ or $\syn_\bold{H}(\mathcal D)=\langle \bold{p}\rangle_{\F_{q^m}}$ for some $P=\langle \bold{p}\rangle \in\mathcal E_{S,T}$. In the first case, $\mathcal D\subseteq K$. In the second case, \[ \mathcal D\subseteq \syn_\bold{H}^{-1}\bigl(\langle \bold{p}\rangle_{\F_{q^m}}\bigr) = \mathcal D_P. \] Conversely, for every $P\in\mathcal E_{S,T}$, we have $\mathcal D_P\subseteq\mathcal C_{S,T}$. Moreover, $K\subsetneq\mathcal D_P$, because \[ \dim_{\F_{q^m}}\mathcal D_P=N-\delta \quad\text{and}\quad \dim_{\F_{q^m}}K=N-\delta-1. \] It follows that the maximal $\F_{q^m}$-linear subspaces contained in $\mathcal C_{S,T}$ are precisely the spaces $\mathcal D_P$. \end{proof}

\begin{thm}\label{thm:DP-MSRD}
Let $\bold{H}$ be the syndrome matrix defined in \eqref{H} and let $\F_q^*=S\,\dot\cup\,T$ be a
partition satisfying Condition~\eqref{eq:switch-conditions}.
Then the following hold.
\begin{enumerate}
\item [$(i)$]
For every $P\in\PG(\delta,q^m)$, the code $\mathcal D_P$ is an $\F_{q^m}$-linear MSRD code with $d_{\rm SR}(\mathcal D_P)=\delta+1$ if and only if $P\notin\Omega_{\delta-1}(\mathcal U_\Lambda)$.

\item [$(ii)$]
The code $\mathcal D_P$ is MSRD for every

\[P\in\{A,B\}\cup\bigcup_{b\in S}\mathcal X_b.
\]
\item [$(iii)$]
For every $a\in\F_q^*\setminus\Pi_\delta(\Lambda)$,
the code $\mathcal D_P$ is MSRD for every
$P\in\mathcal J_a$.
\end{enumerate}
Consequently, if $S\neq\emptyset$, then $\mathcal C_{S,T}$
contains at least
\[
2+
\bigl(|S|+|T\setminus\Pi_\delta(\Lambda)|\bigr)
\frac{q^m-1}{q-1}
\]
maximal $\F_{q^m}$-linear subcodes that are MSRD codes of
dimension $N-\delta$ and minimum sum-rank distance $\delta+1$.
\end{thm}

\begin{proof}
Let $P\in\PG(\delta,q^m)$. Since
$\dim_{\F_{q^m}}\mathcal D_P=N-\delta$, the Singleton-like bound
gives $d_{\rm SR}(\mathcal D_P)\leq\delta+1$, so $\mathcal D_P$
is MSRD if and only if it has no nonzero codeword of sum-rank
weight at most $\delta$.
Let $\bold{0}\neq\bold{z}\in\mathcal D_P$.
If $\bold{z}\in K$, then
$\operatorname{wt}_{\rm SR}(\bold{z})\geq\delta+2$
by Lemma~\ref{lem:kernel-distance}.
If $\bold{z}\notin K$ and
$\operatorname{wt}_{\rm SR}(\bold{z})\leq\delta$, then
Lemma~\ref{lem:syndrome}, applied with $r=\delta$, gives
\[
P=\langle\syn_{\bold{H}}(\bold{z})\rangle
\in\Omega_{\delta-1}(\mathcal U_\Lambda).
\]
Conversely, if $P\in\Omega_{\delta-1}(\mathcal U_\Lambda)$,
Lemma~\ref{lem:syndrome} gives $\bold{z}\notin K$ with
$\langle\syn_{\bold{H}}(\bold{z})\rangle=P$ and
$\operatorname{wt}_{\rm SR}(\bold{z})\leq\delta$.
Such a vector lies in $\mathcal D_P$, so $\mathcal D_P$
is not MSRD. This proves $(i)$.

To prove $(ii)$, it is enough to show that
\[
\Bigl(\{A,B\}\cup\bigcup_{b\in S}\mathcal X_b\Bigr)
\cap\Omega_{\delta-1}(\mathcal U_\Lambda)=\emptyset.
\]
Let $P\in\Omega_{\delta-1}(\mathcal U_\Lambda)$, say $P\in
\mathbb P(\langle W_1,\ldots,W_\ell\rangle_{\F_{q^m}})$
with $W_i\leq_{\F_q}U_{\lambda_i}$ and
$\sum_i\dim_{\F_q}W_i\leq\delta$.
Since $\dim_{\F_q}U_{\lambda_i}=n_i$ and $\sum_i n_i=N>\delta$,
we may enlarge the $W_i$ so that
$\sum_i r_i=\delta$, where $r_i=\dim_{\F_q}W_i$.
By Lemma~\ref{lem:independent-spans}, the space
$W=\langle W_1,\ldots,W_\ell\rangle_{\F_{q^m}}$
has dimension $\delta$, so $\mathbb P(W)$ is a hyperplane
\[
\mathcal H_f:
h_0X_0+\cdots+h_\delta X_\delta=0.
\]
Put $f(X)=\sum_{j=0}^{\delta}h_jX^j$.
Point~3 of Lemma~\ref{lem:roots-curve} gives
$d_{\lambda_i}(f)\geq r_i$, hence
$\sum_i d_{\lambda_i}(f)\geq\delta$.
We distinguish the following cases.

If $P=P(t)\in\mathcal X_b$ with $b\in S$, then
$b\notin\Lambda$ because $\Lambda\subseteq T$.
Point~1 of Lemma~\ref{lem:roots-curve} gives $f(t)=0$,
and hence $d_b(f)\geq1$.
Proposition~\ref{prop:rootspace}, applied to
$\Xi=\Lambda\cup\{b\}$, then yields
$\deg(f)\geq\delta+1$, contradicting $\deg(f)\leq\delta$.

If $P=A$, then $h_\delta=0$ and $\deg(f)\leq\delta-1$,
whereas
\[
\deg(f)\geq\sum_i d_{\lambda_i}(f)\geq\delta,
\]
a contradiction.\\
If $P=B$, then $h_0=0$, so $f=gX$ with
$\deg(g)\leq\delta-1$.
By \eqref{eq:right-X} and Proposition~\ref{prop:rootspace},
\[
\deg(g)
\geq\sum_i d_{\lambda_i}(g)
=\sum_i d_{\lambda_i}(f)
\geq\delta,
\]
again a contradiction.
Thus $(ii)$ follows from $(i)$.

To prove $(iii)$, let $a\in\F_q^*$ and suppose that
\[
P=(1,0,\ldots,0,(-1)^{\delta+1}u)
\in\mathcal J_a\cap\Omega_{\delta-1}(\mathcal U_\Lambda),
\qquad N_m(u)=a.
\]
Choose $W_i$ and $\mathcal H_f$ as above. Since
$P\in\mathcal H_f$, we have $h_0+(-1)^{\delta+1}uh_\delta=0.$
If $h_\delta=0$, then $h_0=0$, so $f=gX$ with
$\deg(g)\leq\delta-2$ and, as above, $\deg(g)\geq\delta$,
a contradiction.
Hence $h_0h_\delta\neq0$ and $\deg(f)=\delta$.
Therefore
\[
\delta=\deg(f)
\geq\sum_i d_{\lambda_i}(f)
\geq\delta,
\]
which forces $d_{\lambda_i}(f)=r_i$ for every $i$ and
$\deg(f)=\sum_{\lambda\in\Lambda}d_\lambda(f)$.
Since $u=(-1)^\delta h_0/h_\delta$,
Proposition~\ref{prop:viete}, applied to $\Xi=\Lambda$, gives
\[
\begin{aligned}
a=N_m(u)
&=N_m\bigl((-1)^\delta\bigr)
  N_m\!\left(\frac{h_0}{h_\delta}\right)\\
&=(-1)^{\delta m}(-1)^{\delta m}
  \prod_{i=1}^{\ell}\lambda_i^{r_i}\\
&=\prod_{i=1}^{\ell}\lambda_i^{r_i}
\in\Pi_\delta(\Lambda),
\end{aligned}
\]
since $0\leq r_i\leq n_i$ and $\sum_i r_i=\delta$.
Thus $\mathcal J_a\cap\Omega_{\delta-1}(\mathcal U_\Lambda)
=\emptyset$ whenever $a\notin\Pi_\delta(\Lambda)$.
Assertion $(iii)$ now follows from $(i)$.

Finally, suppose that $S\neq\emptyset$.
By Theorem~\ref{thm:maximal-subcodes}, the codes $\mathcal D_P$,
with $P\in\mathcal E_{S,T}$, are pairwise distinct maximal
$\F_{q^m}$-linear subcodes of $\mathcal C_{S,T}$.
By $(ii)$ and $(iii)$, those indexed by
\[
\{A,B\}
\cup\bigcup_{b\in S}\mathcal X_b
\cup\bigcup_{a\in T\setminus\Pi_\delta(\Lambda)}\mathcal J_a
\]
are MSRD.
Since these components are pairwise disjoint and each
$\mathcal X_b$ and $\mathcal J_a$ has $(q^m-1)/(q-1)$ points,
their total number of points is
\[
2+
\bigl(|S|+|T\setminus\Pi_\delta(\Lambda)|\bigr)
\frac{q^m-1}{q-1}.
\]
This proves the final assertion.
\end{proof}

\begin{remark}\label{rem:exceptional-classes} Let $G= \langle \Pi_{\delta-1}(\Lambda) \rangle $, we have \[ \Pi_\delta(\Lambda) \subseteq \Lambda G = \lambda_1G \subseteq T. \] Thus all the possible exceptional classes in Theorem~\ref{thm:DP-MSRD} belong to $T$, and \[ T\cap\Pi_\delta(\Lambda)=\Pi_\delta(\Lambda)\neq\emptyset. \] 
\end{remark}
\medskip
Let $\varphi:\F_{q^m}^{N_1}\longrightarrow \F_{q^m}^{N_2}$ be a semilinear map such that $\varphi(K_1)\subseteq K_2$, where $K_1$ and $K_2$ are $\F_{q^m}$-vector subspaces of $\F_{q^m}^{N_1}$ and $\F_{q^m}^{N_2}$, respectively. Then $\varphi$ induces a  semilinear map \[ \overline{\varphi}:\F_{q^m}^{N_1}/K_1\longrightarrow \F_{q^m}^{N_2}/K_2, \qquad \overline{\varphi}(\bold{v}+K_1)=\varphi(\bold{v})+K_2. \]

\begin{thm}\label{cor:geometric-invariant}
For $j\in\{1,2\}$, let $S_j,T_j\subseteq\F_q^*$ satisfy Condition~$(\diamond)$. Let $\bold{H}_j$ be the
corresponding syndrome matrix as in \eqref{H}, and put
$K_j=\ker(\syn_{\bold{H}_j})$ and $\mathcal E_j=\mathcal E_{S_j,T_j}$, so
that $\mathcal C_j=\syn_{\bold{H}_j}^{-1}(\hat{\mathcal E}_j)$.\\
Assume $S_j \neq\emptyset$, $j \in \{1,2\}$. If $\varphi$ is an $\F_{q^m}$-semilinear sum-rank
isometry of $\F_{q^m}^{N}$ with $\varphi(\mathcal C_1)=\mathcal C_2$, then
$\varphi(K_1)=K_2$, and
\[
  \psi=\overline{\syn}_{\bold{H}_2}\circ\overline{\varphi}\circ
  \overline{\syn}_{\bold{H}_1}^{-1}
\]
is an invertible $\F_{q^m}$-semilinear map of $\F_{q^m}^{\delta+1}$ such
that $\psi(\hat{\mathcal E}_1)=\hat{\mathcal E}_2$. In particular,
$\mathcal E_1$ and $\mathcal E_2$ are equivalent under
$\mathrm{P}\Gamma\mathrm{L}(\delta+1,q^m)$.
\end{thm}

\begin{proof}
By Theorem~\ref{thm:maximal-subcodes}, the maximal $\F_{q^m}$-linear
subspaces contained in $\mathcal C_j$ are the spaces
$\syn_{\bold{H}_j}^{-1}\bigl(\langle\bold{p}\rangle_{\F_{q^m}}\bigr)$, with
$\langle\bold{p}\rangle\in\mathcal E_j$, and their intersection is $K_j$.
Since $\varphi$ and $\varphi^{-1}$ are $\F_{q^m}$-semilinear, they map
$\F_{q^m}$-subspaces to $\F_{q^m}$-subspaces and preserve inclusions. Hence
$\varphi(\mathcal C_1)=\mathcal C_2$ implies that $\varphi$ maps the maximal
$\F_{q^m}$-linear subspaces of $\mathcal C_1$ onto those of $\mathcal C_2$,
and therefore $\varphi(K_1)=K_2$. Thus $\overline{\varphi}$ is a
well-defined semilinear isomorphism. Since $\bold{H}_j$ has full rank,
$\syn_{\bold{H}_j}$ is surjective with kernel $K_j$, so
$\overline{\syn}_{\bold{H}_j}$ is an $\F_{q^m}$-linear isomorphism. It
follows that $\psi=\psi=\overline{\syn}_{\bold{H}_2}\circ\overline{\varphi}\circ
  \overline{\syn}_{\bold{H}_1}^{-1}$ is an invertible semilinear map. Finally,
$\mathcal C_j$ is a union of cosets of $K_j$ with
$\overline{\syn}_{\bold{H}_j}(\mathcal C_j/K_j)=\hat{\mathcal E}_j$. Since
$\overline{\varphi}(\mathcal C_1/K_1)=\mathcal C_2/K_2$, we get
$\psi(\hat{\mathcal E}_1)=\hat{\mathcal E}_2$. Passing to projective points
gives the last assertion.
\end{proof}

\begin{remark} Theorem~\ref{cor:geometric-invariant} shows that the projective equivalence class of the switched set $\mathcal E_{S,T}$ is an invariant of the code under $\F_{q^m}$-semilinear sum-rank equivalence. In particular, if $S\neq\emptyset$, then $\mathcal E_{S,T}$ is not a line by Corollary~\ref{cor:no-line}, whereas $\mathcal E_{\emptyset,\F_q^*}=AB.$ Consequently, no code of the family with $S\neq\emptyset$ is $\F_{q^m}$-semilinearly sum-rank equivalent to $\mathcal C_{\emptyset,\F_q^*}$. Since $\mathcal C_{\emptyset,\F_q^*}$ is a linearized Reed--Solomon code, this gives a geometric distinction between the non-additive members of the family and the linear ones. More generally, the multiset \[ \left\{\!\{ |\mathcal E_{S,T}\cap\ell|: \ell\text{ is a line of }\PG(\delta,q^m) \right\}\!\} \] is an invariant of $\mathcal C_{S,T}$. In particular, this multiset contains the intersection number \[ |\mathcal E_{S,T}\cap AB| = 2+|T|\frac{q^m-1}{q-1}. \] Such intersection may be used to distinguish members of the family from one another and from other constructions. \end{remark}

For a subset $\mathcal S\subseteq\F_{q^m}^{N}$ and an integer $w\geq 0$, we set
\[
  A_w(\mathcal S)=\big|\{\bold{x}\in\mathcal S \,:\, \mathrm{wt}_{\rm SR}(x)=w\}\big|.
\]

Let $1\leq k\leq N$ and let
$\mathcal C\subseteq\F_{q^m}^{N}$ be an $\F_{q^m}$-linear MSRD code
of dimension $k$, with block lengths $n_1,\ldots,n_\ell$.
The sum-rank weight distribution of $\F_{q^m}$-linear MSRD codes
has been completely determined in~\cite[Theorem~VI.5 and
Remark~VI.11]{BGR}. In particular, $A_w(\mathcal C)$ depends
only on the parameters $q,m,n_1,\ldots,n_\ell, k,w$.
We denote this value by $\mathcal A_w(k)$. For $k=0$, we set
\[
\mathcal A_w(0)=A_w(\{\bold{0}\}),
\]
that is,
\[
\mathcal A_0(0)=1,
\qquad
\mathcal A_w(0)=0
\quad\text{for }w>0.
\]

\begin{lem}\label{lem:switch-weight-balance}
Let $\bold{H}$ be the syndrome matrix defined in \eqref{H} and let $\F_q^*=S\,\dot\cup\,T$ be a
partition satisfying Condition~\eqref{eq:switch-conditions}.
For every $b\in S$ and every integer $0\leq w\leq N$, one has
\begin{equation}\label{eq:switch-weight-balance}
\sum_{P\in\mathcal X_b} A_w(\mathcal D_P)
=
\sum_{P\in\mathcal J_b} A_w(\mathcal D_P)
=
\frac{q^m-1}{q-1}\,\mathcal A_w(N-\delta),
\end{equation}
\end{lem}

\begin{proof}
Fix $b\in S$. By Remark~\ref{rem:exceptional-classes},
$\Pi_\delta(\Lambda)\subseteq T$.
Since $S\cap T=\emptyset$, it follows that
$b\notin\Pi_\delta(\Lambda).$

By Theorem~\ref{thm:DP-MSRD}$(ii)$, for every
$P\in\mathcal X_b$, the code $\mathcal D_P$ is an
$\F_{q^m}$-linear MSRD code of dimension $N-\delta$.
Moreover, since $b\notin\Pi_\delta(\Lambda)$,
Theorem~\ref{thm:DP-MSRD}$(iii)$ yields the same conclusion for every
$P\in\mathcal J_b$. Hence all the codes $\mathcal D_P$ with 
$P\in\mathcal X_b\cup\mathcal J_b$,
are $\F_{q^m}$-linear MSRD codes of dimension $N-\delta$.

By \cite[Theorem~VI.5 and Remark~VI.11]{BGR}, the sum-rank weight
distribution of an MSRD code is determined by its parameters.
Therefore,  
\[
A_w(\mathcal D_P)=\mathcal A_w(N-\delta)
\]
for every
$P\in\mathcal X_b\cup\mathcal J_b$.

Finally, since 
\[
|\mathcal X_b|
=
|\mathcal J_b|
=
\frac{q^m-1}{q-1},
\]
we have 
\[
\sum_{P\in\mathcal X_b} A_w(\mathcal D_P)
=
\sum_{P\in\mathcal J_b} A_w(\mathcal D_P)
=
\frac{q^m-1}{q-1}\,\mathcal A_w(N-\delta).
\]
\end{proof}

\begin{thm}\label{thm:weight-invariance}
Let $\bold{H}$ be the syndrome matrix defined in \eqref{H} and let $\F_q^*=S\,\dot\cup\,T$ be a
partition satisfying Condition~\eqref{eq:switch-conditions}.
Then, for every integer $0\leq w\leq N$,
\begin{equation}\label{eq:weight-invariance}
A_w(\mathcal C_{S,T})
=
A_w(\mathcal C_{\emptyset,\F_q^*})
=
\mathcal A_w(N-\delta+1).
\end{equation}
Consequently, the sum-rank weight distribution of
$\mathcal C_{S,T}$ depends only on the parameters $q,\ m,\ n_1,\ldots,n_\ell, \ \delta$.
In particular, the number of codewords of minimum sum-rank weight is
\begin{equation}\label{eq:minimum-weight-count}
A_\delta(\mathcal C_{S,T})
=
(q^m-1)
\sum_{\substack{
r_1+\cdots+r_\ell=\delta\\
0\leq r_i\leq n_i}}
\prod_{i=1}^{\ell}
\genfrac{[}{]}{0pt}{}{n_i}{r_i}_q.
\end{equation}
\end{thm}

\begin{proof}
Let
$K=\ker(\syn_{\bold H})$ and 
consider first the trivial partition
$S= \emptyset$ and $T= \F_q^*$. As seen above, $\mathcal E_{\emptyset,\F_q^*}=AB$,
and 
$\mathcal C_{\emptyset,\F_q^*}
=
\syn_{\bold H}^{-1}
\bigl(
\langle\boldsymbol e_0,\boldsymbol e_\delta
\rangle_{\F_{q^m}}
\bigr)$ is MSRD. Hence, by
\cite[Theorem~VI.5 and Remark~VI.11]{BGR},
\begin{equation}\label{eq:C0-weight-distribution}
A_w(\mathcal C_{\emptyset,\F_q^*})
=
\mathcal A_w(N-\delta+1)
\end{equation}
for every $0\leq w\leq N$.

Now let $(S,T) \neq (\emptyset, \F_q^*)$.
For a projective point set
$\mathcal E\subseteq\PG(\delta,q^m)$,
let $$\hat{\mathcal E}=\{\alpha \bold{v} \in \F_{q^m}^{\delta+1} \colon  \alpha \in \F_{q^m}, \langle \bold{v} \rangle \in \mathcal{E} \}$$ 
and put $\mathcal C(\mathcal E)
=
\syn_{\bold H}^{-1}(\hat{\mathcal E}).$
By Remark~\ref{rm:dec},
\[
\mathcal C(\mathcal E)
=
K\,\dot\cup\!
\bigcup_{P\in\mathcal E}
(\mathcal D_P\setminus K).
\]
and hence
\begin{equation}\label{eq:general-weight-decomposition}
A_w(\mathcal C(\mathcal E))
=
\sum_{P\in\mathcal E}A_w(\mathcal D_P)
-
q^m A_w(K).
\end{equation}

By construction, the point set  $\mathcal E_{S,T}$ is obtained from $AB$ by replacing
$\mathcal J_b$ with $\mathcal X_b$ for every $b\in S$.

Subtracting the two instances of
\eqref{eq:general-weight-decomposition}, we obtain
\[
\begin{aligned}
&A_w(\mathcal C_{S,T})
-
A_w(\mathcal C_{\emptyset,\F_q^*})
\\
&\qquad =
\sum_{b\in S}
\left(
\sum_{P\in\mathcal X_b}A_w(\mathcal D_P)
-
\sum_{P\in\mathcal J_b}A_w(\mathcal D_P)
\right).
\end{aligned}
\]
By Lemma~\ref{lem:switch-weight-balance}, every term in parentheses
vanishes. Therefore,
\[
A_w(\mathcal C_{S,T})
=
A_w(\mathcal C_{\emptyset,\F_q^*})
\]
for every $0\leq w\leq N$.
Together with \eqref{eq:C0-weight-distribution}, this proves
\eqref{eq:weight-invariance}.
Since $\mathcal C_{\emptyset,\F_q^*}$ is an $\F_{q^m}$-linear MSRD code
of minimum sum-rank distance $\delta$, its support distribution, and hence
its sum-rank weight distribution, is completely determined by its
parameters; see \cite[Theorem~VI.5 and Remark~VI.11]{BGR}. In particular,
the number of codewords of minimum sum-rank weight is as \eqref{eq:minimum-weight-count}.
\end{proof}

\section{The case \texorpdfstring{$\ell=1$}{ell=1}}
\label{sec:rank-recovery}

Assume that $\ell=1$, $\Lambda=\{1\}$ and $a_1=1$. Let $N=n_1=m$ and $2\leq\delta\leq m-1$.
Since $\Pi_{\delta-1}(\{1\})=\{1\}$, the switching conditions reduce
to $1\in T$, where $S=\F_q^*\setminus T$. The set
$\mathcal E_{S,T}$ is the Donati--Durante exterior set
$\mathcal E_T$ in \eqref{eq:ET}, with $d=\delta$.

Let $(\beta_1,\ldots,\beta_m)$ be the $\F_q$-basis used to construct
$\bold{H}$ (cf. \eqref{H}). Our choice $a_1=1$ gives the Moore matrix
\[
 \bold{H}=
 \begin{pmatrix}
  \beta_1&\cdots&\beta_m\\
  \beta_1^\sigma&\cdots&\beta_m^\sigma\\
  \vdots&\ddots&\vdots\\
  \beta_1^{\sigma^\delta}&\cdots&\beta_m^{\sigma^\delta}
 \end{pmatrix}.
\]
The following result identifies the entire syndrome preimage with the
square cone code recalled in Section \ref{DonDurcode}.

Let $\sigma$ be a generator of $\mathrm{Gal}(\F_{q^m} \vert \F_q)$ and denote by
\[
\mathcal{L}_{m,q,\sigma}[X]
=
\left\{
\sum_{i=0}^{m-1} f_i X^{\sigma^i}
\colon f_i \in \F_{q^m}
\right\}
\]
the $\F_{q^m}$-vector space of $\sigma$-linearized polynomials of
$\sigma$-degree less than $m$.
Evaluation on $\F_{q^m}$ induces an isomorphism of $\F_{q^m}$-vector spaces
\[
\mathcal{L}_{m,q,\sigma}[X]
\longrightarrow \mathrm{End}_{\F_q}(\F_{q^m}),
\qquad f\longmapsto\bigl(x\longmapsto f(x)\bigr).
\]
Clearly, we use this identification to endow $\mathcal{L}_{m,q,\sigma}[X]$ with
the rank distance
\[
d_{\rm R}(f,g)=\operatorname{rk}_{\F_q}(f-g),
\]
where the rank is that of the associated $\F_q$-linear map, see \cite[Section~3.4]{LN}.

\begin{prop}\label{prop:rank-recovery}
Let $k=m-\delta+1$, and let $\mathcal C_{\sigma,T}$ be the code defined
in \eqref{eq:CsigmaT-geometric}, with $d=\delta$. The map
\[
\begin{split}
\Psi:\F_{q^m}^m&\longrightarrow\mathcal{L}_{m,q,\sigma}[X],\\
\bold{z}=(z_1,\ldots,z_m)&\longmapsto
\sum_{i=0}^{m-1}
\left(\sum_{j=1}^m z_j\beta_j^{\sigma^i}\right)X^{\sigma^i}
\end{split}
\]
is an $\F_{q^m}$-linear rank-metric isometry satisfying
\[
\Psi(\mathcal C_{S,T})=\mathcal C_{\sigma,T}.
\]
Thus the full-block construction with $\ell=1$, $\Lambda=\{1\}$ and
$a_1=1$ recovers the square MRD codes in
Section~\ref{DonDurcode}.
\end{prop}

\begin{proof}
The map $\Psi$ is $\F_{q^m}$-linear by definition. Since $\sigma$
generates $\operatorname{Gal}(\F_{q^m}\vert\F_q)$, the field trace
satisfies
\[
\operatorname{Tr}_{q^m/q}(y)=\sum_{i=0}^{m-1}y^{\sigma^i}
\qquad\text{for every }y\in\F_{q^m}.
\]
Consequently, for every $\bold{z}\in\F_{q^m}^m$ and $x\in\F_{q^m}$,
\[
\Psi(\bold{z})(x)
=\sum_{j=1}^m z_j\operatorname{Tr}_{q^m/q}(\beta_jx).
\]
Let $(\beta_1^*,\ldots,\beta_m^*)$ be the trace-dual basis of
$(\beta_1,\ldots,\beta_m)$.
It follows that
\[
\Psi(\bold{z})(\beta_h^*)=z_h
\qquad(1\leq h\leq m).
\]
Thus $\Psi$ is injective and hence bijective, since both its domain
and codomain have dimension $m$ over $\F_{q^m}$. Its inverse is
\[
\Psi^{-1}(f)=\bigl(f(\beta_1^*),\ldots,f(\beta_m^*)\bigr).
\]
The trace formula shows that
$\operatorname{im}\Psi(\bold{z})\subseteq
\langle z_1,\ldots,z_m\rangle_{\F_q}$, while evaluation on the
trace-dual basis gives the reverse inclusion. Therefore,
\[
\operatorname{rk}_{\F_q}\Psi(\bold{z})
=\dim_{\F_q}\langle z_1,\ldots,z_m\rangle_{\F_q}
=\operatorname{wt}_{\rm SR}(\bold{z}).
\]
Since $\ell=1$, the sum-rank weight is the rank weight. By linearity,
$\Psi$ preserves rank distances.

Now write
\[
\Psi(\bold{z})=\sum_{i=0}^{m-1}c_iX^{\sigma^i},
\qquad c_i=\sum_{j=1}^m z_j\beta_j^{\sigma^i}.
\]
By the definition of $\bold{H}$, we have 
\[
(c_0,\ldots,c_\delta)^{\mathsf T}
=\bold{H}\bold{z}^{\mathsf T}
=\syn_{\bold{H}}(\bold{z}).
\]
Moreover, 
\[
\bold{z}\in\mathcal C_{S,T}
\quad\Longleftrightarrow\quad
(c_0,\ldots,c_\delta)\in\hat{\mathcal E}.
\]
and 
\[
\Psi(\mathcal C_{S,T})
=\left\{
\sum_{i=0}^{m-1}c_iX^{\sigma^i}
\colon c_i\in\F_{q^m},\quad
(c_0,\ldots,c_\delta)\in\hat{\mathcal E}
\right\}.
\]

Under the coordinate embedding obtained by appending $m-\delta-1$
zeros, $\mathcal E_{S,T}$ is identified with the set $\mathcal E_T$
in \eqref{eq:ET}. By the definition of the projective cone
$\mathbb K_T=\mathbb K(\Lambda^\star,\mathcal E_T)$, a nonzero vector
$\bold{c}=(c_0,\ldots,c_{m-1})$ represents a point of $\mathbb K_T$
if and only if
\[
(c_0,\ldots,c_\delta)=\bold{0}
\quad\text{or}\quad
\langle(c_0,\ldots,c_\delta)\rangle\in\mathcal E_{S,T}.
\]
The first case gives the vertex $\Lambda^\star$; in the second case,
the unrestricted last coordinates give the points joining the base
point to the vertex. This condition is equivalent to
$(c_0,\ldots,c_\delta)\in\hat{\mathcal E}$. Taking all nonzero
representatives of $\mathbb K_T$, together with the zero vector,
and using \eqref{eq:CsigmaT-geometric}, we obtain
\[
\Psi(\mathcal C_{S,T})=\mathcal C_{\sigma,T}.
\]
When $\delta=m-1$, there are no remaining coefficients and the vertex
$\Lambda^\star$ is empty. In this case $\mathbb K_T=\mathcal E_T$,
and the same argument applies.
\end{proof}

\section*{Acknowledgements}

The authors are members of the Gruppo Nazionale per le Strutture
Algebriche, Geometriche e le loro Applicazioni (GNSAGA) of the Istituto
Nazionale di Alta Matematica (INdAM). This research was partially
supported by INdAM--GNSAGA.

\end{document}